%% file: ISIT_2021_arXiv_v2.tex
\documentclass[conference,10pt,letterpaper]{IEEEtran}

\include{preamble_global}

\usepackage{subfigure}
\usepackage{multirow}
\usepackage{array}
\usepackage{amsmath}
\usepackage[table,dvipsnames]{xcolor}
\usepackage{enumitem}
\usepackage{stfloats}
\usepackage{cite}
\usepackage[toc,acronym,nonumberlist,nopostdot]{glossaries}
\makeindex
\makenoidxglossaries
\usepackage{array,makecell,booktabs}
\usepackage{tabstackengine}
\newcolumntype{M}[1]{>{\centering\arraybackslash}m{#1}}
\usepackage{multicol,multirow,pgf}
\usepackage{tikz,pgfplots}
\usetikzlibrary{shapes}
\usetikzlibrary{plotmarks}
\usetikzlibrary{spy}
\usepackage{pgfplots,pgf}
\pgfplotsset{minor grid style={dotted,gray!40}}
\pgfplotsset{major grid style={dashed,gray!40}}

\allowdisplaybreaks

\renewcommand{\rv}[1]{{\mathsf{#1}}}
\newcommand{\rvVec}[1]{{\bm{\mathsf{#1}}}}

\newcommand{\ind}[1]{{\mathbbm{1}{\{#1\}}}}

\renewcommand{\defeq}{=}

\renewcommand{\Id}{\mat{I}}

\newcommand{\EbNo}{{E_{\rm b}/N_0}}

\newcommand{\md}{\frac{|\Wc_0|}{|\Wc|}}
\newcommand{\fa}{\frac{|\Wc_0'|}{|\widehat{\Wc}|} }

\newacronym{MAC}{MAC}{multiple access channel}
\newacronym{RAC}{MAC}{random-access channel}
\newacronym{UMRA}{UMRA}{unsourced massive random-access}
\newacronym{SIMO}{SIMO}{single-input multiple-output}
\newacronym{SISO}{SISO}{single-input single-output}
\newacronym{iid}{i.i.d.}{independent and identically distributed}
\newacronym{ML}{ML}{maximum likelihood}
\newacronym{PEP}{PEP}{pair-wise error probability}
\newacronym{LLR}{LLR}{log-likelihood ratio}
\newacronym{SNR}{SNR}{signal-to-noise ratio}
\newacronym{RCB}{RCB}{random-coding bound}
\newacronym{MD}{MD}{misdetection}
\newacronym{FA}{FA}{false alarm}
\newacronym{PMF}{PMF}{probability mass function}
\newacronym{wlog}{w.l.o.g.}{without loss of generality}
\newacronym{wrt}{w.r.t.}{with respect to}
\newacronym{SA}{SA}{slotted {ALOHA}}
\newacronym{DoF}{DoF}{degrees of freedom}
\newacronym{rDoF}{rDoF}{real degrees of freedom}
\newacronym{ROC}{ROC}{receiver operating characteristic}
\newacronym{IoT}{IoT}{Internet of Things}
\newacronym{TIN}{TIN}{treating-interference-as-noise}
\newacronym{PDF}{PDF}{probability density function}

\title{Massive Uncoordinated Access \\With Random User Activity} 

\author{\IEEEauthorblockN{Khac-Hoang Ngo, Alejandro Lancho, Giuseppe Durisi, and Alexandre Graell i Amat} 
	\IEEEauthorblockA{Department of Electrical Engineering, Chalmers University of Technology, 41296 Gothenburg, Sweden}
}

\newcommand{\revise}[1]{#1} 
\newcommand{\revisee}[1]{#1} 

\begin{document}
	
\maketitle
\date{\today}
\begin{abstract}
	We extend the seminal work by Polyanskiy (2017) on massive uncoordinated access to the case where the number of active users is random and unknown {\em a priori}. We define a random-access code accounting for both misdetection~(MD) and false-alarm~(FA), and derive a random-coding achievability bound for the Gaussian multiple access channel. Our bound captures the fundamental trade-off between MD and FA probabilities. 
	\revise{It} suggests that 
	lack of knowledge of the number of active users entails a small penalty in power efficiency. \revise{For a typical scenario, to achieve both MD and FA probabilities below $0.1$, the required  energy per bit \revisee{predicted by} our bound is $0.5$--$0.7$ dB higher than that \revisee{predicted by} the bound in Polyanskiy (2017) 
		for a known number of active users.} Taking both MD and FA into account, we \revisee{use our bound to benchmark the energy efficiency of some recently proposed massive random access schemes.}
\end{abstract}


\section{Introduction} \label{sec:intro}
Under the paradigm of the \gls{IoT}, the number of connecting devices is increasing dramatically. IoT devices are mostly battery limited and transmit short packets in a sporadic and uncoordinated manner~\cite{Chen2020_massiveAccess,Wu2020_massiveAccess}. This calls for new theoretical frameworks that help to understand the fundamental limits of massive random access and provide guidelines for system design. 
Polyanskiy~\cite{PolyanskiyISIT2017massive_random_access} proposed a novel formulation for the massive uncoordinated access problem with three key assumptions: i) all users employ a common codebook and the decoder only aims to return a list of messages without recovering users' identities; ii) the error event is defined per user 
and the error probability is averaged over the users; iii) each user sends a fixed amount of information bits within a finite frame length. Under this formulation, traditional as well as novel random access protocols \cite{Berioli2016NOW} yield achievability bounds. In \cite{PolyanskiyISIT2017massive_random_access}, an achievability bound for the Gaussian \gls{MAC} was derived and it was shown that modern random access schemes exhibit a large gap to this bound. This gap has been later reduced in, e.g., \cite{Ordentlich2017low_complexity_random_access,Vem2019,Fengler2019sparcs,Amalladinne2020unsourced,Amalladinne2020,Pradhan2020}. Polyanskiy's framework has been extended to the quasi-static fading channel~\cite{Kowshik2020}, multiple-antenna channel~\cite{Fengler2019nonBayesian}, \revisee{and a setup with common alarm messages~\cite{Stern2019}.}

In Polyanskiy's achievability bound, the number of active users is fixed and known to the receiver, an assumption that has practical shortcomings. Since \gls{IoT} devices access the channel at random times and in a grant-free manner, the number of active users varies over time, and hence, it is typically unknown to the receiver. Therefore, the bound in \cite{PolyanskiyISIT2017massive_random_access} may be an overoptimistic benchmark for random-access schemes that are designed to work with unknown number of active users. Moreover, when the number of active users is unknown, the decoder needs to determine the list size. Choosing a list size smaller than the number of active users will result in \glspl{MD}\textemdash i.e., transmitted messages that are not included in the decoded list\textemdash whereas choosing it larger than the number of active users will result in \glspl{FA}\textemdash i.e., decoded messages that were not transmitted. Furthermore, additional \glspl{MD} and \glspl{FA} may occur in the decoding process. 
There is a trade-off between \gls{MD} and \gls{FA} probabilities. A decoder that always outputs the whole codebook will never misdetect, but has \gls{FA} probability close to one; similarly, a decoder that always outputs an empty set will never raise a \gls{FA} but always misdetects. 
Characterizing the \gls{MD}--\gls{FA} trade-off is a fundamental engineering challenge that was not addressed in \cite{PolyanskiyISIT2017massive_random_access}. An achievability bound for the Gaussian \gls{MAC} with unknown number of active users was presented in \cite{Effros2018ISIT}. However, the authors consider the joint-user error event instead of the per-user error event, and thus, \gls{MD} and \gls{FA} are not explicitly considered. In short, a random-coding bound accounting for both \gls{MD} and \gls{FA}, which can serve as a benchmark for common-codebook massive uncoordinated random access with random user activity, is still missing.


Most of the practical algorithms that have been proposed for common-codebook massive random access require knowledge of the number of active users. 
Advanced ALOHA schemes, such as irregular repetition slotted ALOHA~(IRSA)~\cite{Liva2011IRSA}, 
can also operate when the number of active users is unknown. However, research on modern random access protocols~\cite{Berioli2016NOW}, such as IRSA, has traditionally focused on characterizing and minimizing the packet loss rate, which accounts only for \gls{MD}. The scheme proposed in \cite{Vem2019} also addressed \gls{MD} only.  Minimizing the \gls{MD} probability alone can entail a high \gls{FA} probability. In~\cite{Decurninge2020}, a tensor-based communication scheme was proposed, \revise{and both \gls{MD} and \gls{FA} probabilities are reported in the performance evaluation. Another scheme \revisee{for which} both \gls{MD} and \gls{FA} probabilities \revisee{are reported} was recently proposed in~\cite{fengler2020pilot} for the quasi-static fading \gls{MAC} \revisee{and for the case in which} the receiver has a large number of antennas.}


In this work, we extend Polyanskiy's bound to the case where the number of active users is {\em random} and {\em unknown}. To this end, we first extend the definition of a random-access code provided in~\cite{PolyanskiyISIT2017massive_random_access} to account for both \gls{MD} and \gls{FA} probabilities. Then, we derive a random-coding bound for the Gaussian \gls{MAC}. Unlike the scheme in~\cite{PolyanskiyISIT2017massive_random_access}, our decoder does not assume knowledge of the number of active users, and thus cannot use it to set the decoded list size. Instead, we let our decoder decide the best list size within a predetermined interval around an estimated value of the number of active users. \revisee{Our decoding metric is similar to that used in \cite{Stern2019}. However, different from \cite{Stern2019}, we limit the decoded list size to be in an interval 
	to avoid overfitting.} 

\revisee{Compared with the bound in \cite{PolyanskiyISIT2017massive_random_access}}, our bound suggests that  the lack of knowledge of the number of active users entails a small penalty in power efficiency. Furthermore, \revise{we apply our bound to \revisee{characterize} \gls{MD} and \gls{FA} in slotted ALOHA with multi-packet reception (SA-MPR). Using our bound, we \revisee{benchmark the energy efficiency of} SA-MPR and \revisee{of the} massive random access schemes \revisee{proposed in \cite{Fengler2019sparcs,Amalladinne2020unsourced}}.}
For instance, for a system with $\revise{300}$ active users in average, to achieve both \gls{MD} and \gls{FA} probabilities below $10^{-1}$, the required energy per bit \revisee{predicted by} our achievability bound is \revise{$0.65$~dB higher than that \revisee{predicted by} the bound for a known number of active users \cite{PolyanskiyISIT2017massive_random_access}. In the same setting, the required energy per bit predicted by our bound is $9$~dB, $4.1$~dB, and $3.6$~dB lower than that of the SA-MPR bound, the scheme based on sparse regression code (SPARC)~\cite{Fengler2019sparcs}, and an enhancement of SPARC~\cite{Amalladinne2020unsourced}, respectively.}


\subsubsection*{Notation}
Random quantities are denoted with non-italic letters with sans-serif font, e.g., a scalar $\rv{x}$ and a vector $\rvVec{v}$. 
Deterministic quantities are denoted 
with italic letters, e.g., a scalar $x$ and a vector $\bm{v}$. 
The Euclidean norm is denoted by $\|\cdot\|$. 
We use $\mathfrak{P}(\Ac)$ to denote the set of all subsets of $\Ac$; $[n]$ denotes the set of integers $\{1,\dots,n\}$ if $n \ge 1$ and $[n] \defeq \emptyset$ if $n=0$; $[m:n] \defeq \{m,m+1,\dots,n\}$ if $m \le n$ and $[m:n] \defeq \emptyset$ if $m>n$; $x^+ \defeq \max\{x,0\}$; $\ind{\cdot}$ is the indicator function. The set of natural and complex numbers are denoted by $\mathbb{N}$ and $\CC$, respectively. We denote the Gamma function by $\Gamma(x) \defeq \int_{0}^{\infty}z^{x-1}e^{-z}dz$, and the lower and upper incomplete Gamma functions by $\gamma(x,y) \defeq \int_{0}^{y}z^{x-1}e^{-z}dz$ and $\Gamma(x,y) \defeq \int_{y}^{\infty}z^{x-1}e^{-z}dz$, respectively. 

\section{Random-Access Channel} \label{sec:channel}
We consider a \gls{MAC} in which a random set of $\rv{K}_{\rm a}$ users transmit their messages to a receiver over $n$ uses of a stationary memoryless channel. Let $\rv{x}_k \in \Xc$ be the transmitted signal of user $k$ in a channel use. 
Given $\rv{K}_{\rm a} = K_{\rm a}$, the channel law is given by $P_{\rv{y} \cond \rv{x}_1,\dots,\rv{x}_{{K}_{\rm a}}}$. Thus this random-access channel is characterized by the \gls{PMF} $P_{\rv{K}_{\rm a}}$ of $\rv{K}_{\rm a}$ and by the set of conditional probabilities $\{P_{\rv{y} \cond \rv{x}_1,\dots,\rv{x}_{{K}_{\rm a}}} \colon \Xc^{K_{\rm a}} \to \Yc \}_{K_{\rm a} \in \mathbb{N}}$. 
As in~\cite{PolyanskiyISIT2017massive_random_access}, we assume that the channel law is permutation invariant. 
We further assume that the receiver does not know the realizations of $\rv{K}_{\rm a}$.  

As in~\cite{PolyanskiyISIT2017massive_random_access}, our model differs from the classical \gls{MAC} in that the total number of users is not limited, all users employ the same codebook, and the receiver decodes up to a permutation of messages. 
However, as opposed to~\cite{PolyanskiyISIT2017massive_random_access}, where the number of active users is assumed to be fixed and known, we assume that $\rv{K}_{\rm a}$ is random and unknown. 
We therefore need to account for both \gls{MD} and \gls{FA}.
We next rigorously define the \gls{MD} and \gls{FA} probabilities, as well as the notion of a random-access code.

\begin{definition}[Random-access code] \label{def:code}
	Consider a random-access channel characterized by $\big\{P_{\rv{K}_{\rm a}}, \{P_{\rv{y} \cond \rv{x}_1,\dots,\rv{x}_{{K}_{\rm a}}}\}_{K_{\rm a} \in \mathbb{N}}\big\}$. \revisee{An $(M,n,\epsilon_{\rm MD},\epsilon_{\rm FA})$ random-access code for this channel, where $M$ and $n$ are positive integers and $\epsilon_{\rm MD},\epsilon_{\rm FA} \in (0,1)$, consists of:}
	\begin{itemize}[leftmargin=*]
		\item \revisee{A random variable $\rv{U}$ defined on a set $\Uc$ 
			that is revealed to both the transmitter and the receiver before the start of the transmission. This random variable acts as common randomness and allows for the use of randomized coding strategies.}
		
		\item \revisee{An encoder mapping $f\colon \Uc \times [M] \to \Xc^n$ defining the transmitted codeword $\rvVec{x}_i = f(\rv{U},\rv{w}_i)$ of user $i$ for a given message $\rv{w}_i$, which is assumed to be uniformly distributed over $[M]$.}
		
		\item \revisee{A decoding function $g\colon \Uc \times \Yc^n \to \mathfrak{P}([M])$ providing an estimate $\widehat{\Wc} = \{\hat{\rv{w}}_1,\dots,\hat{\rv{w}}_{|\widehat{\Wc}|}\} = g(\rv{U},\rvVec{y})$ of the list of transmitted messages, where $\rvVec{y} = [\rv{y}(1) \dots \rv{y}(n)]^\T$ denotes the channel output sequence.} 
	\end{itemize}
	Let $\widetilde{\Wc} = \{\widetilde{\rv{w}}_1,\dots,\widetilde{\rv{w}}_{|\widetilde{\Wc}|}\}$ denotes the set of distinct elements of $\{{\rv{w}}_1,\dots,{\rv{w}}_{\rv{K}_{\rm a}}\}$. \revisee{We assume that the decoding function satisfies the following constraints on the \gls{MD} and \gls{FA} probabilities:
		\begin{align}
			\!\!\!P_{\rm MD} &\defeq \E\Bigg[{\ind{|\widetilde{\Wc}| \ne 0} \cdot \frac{1}{|\widetilde{\Wc}|} \sum_{i=1}^{|\widetilde{\Wc}|} \P[\widetilde{\rv{w}}_i \!\notin\! \widehat{\Wc}]}\!\Bigg] \!\le \epsilon_{\rm MD}, \label{eq:def_pMD}\\
			\!\!\!P_{\rm FA} &\defeq \E\Bigg[{\ind{|\widehat{\Wc}| \ne 0} \cdot \frac{1}{|\widehat{\Wc}|} \sum_{i=1}^{|\widehat{\Wc}|} \P[\hat{\rv{w}}_i \notin {\Wc}]}\Bigg] \!\le \epsilon_{\rm FA}, \label{eq:def_pFA}
		\end{align}
		The expectations in \eqref{eq:def_pMD} and \eqref{eq:def_pFA} are with respect to the size of $\Wc$ and $\widehat{\Wc}$, respectively.} 
\end{definition} 	
	
		In the random-access code defined in~\cite{PolyanskiyISIT2017massive_random_access}, the decoder outputs a list of messages of size equal to the number of active users, which is assumed to be known. In such a setup, a \gls{MD} implies a \gls{FA}, and vice versa. Hence, the two types of errors become indistinguishable. In our setup, the decoded list size can be different from the number of transmitted messages. Hence, we introduce explicitly the \gls{MD} and \gls{FA} probabilities. This allows us to characterize the \gls{MD}--\gls{FA} trade-off. 

Hereafter, we consider the Gaussian \gls{MAC} with $\{P_{\rv{y}|\rv{x}_1,\dots,\rv{x}_{{K}_{\rm a}}}\}$ imposed by
$
	\rvVec{y} = \sum_{i=1}^{\rv{K}_{\rm a}}\rvVec{x}_i + \rvVec{z}, 
$
where $\{\rvVec{x}_i\}_{i=1}^{\rv{K}_{\rm a}}$ are the transmitted signals over $n$ channel uses and $\rvVec{z} \sim \Cc\Nc(\mathbf{0},\Id_n)$ is the Gaussian noise 
independent of $\{\rvVec{x}_i\}_{i=1}^{\rv{K}_{\rm a}}$.
We consider the power constraint $\|\rvVec{x}_i\|^2 \le nP, \forall i \in [\rv{K}_{\rm a}]$. 

\section{Random-Coding Bound} \label{sec:RCU}
The random-coding bound in~\cite[Th.~1]{PolyanskiyISIT2017massive_random_access} 
is derived by constructing a random-coding scheme as follows. Let $\Wc = \{\rv{w}_1, \dots, \rv{w}_{K_{\rm a}}\} \subset [M]$ be the set of transmitted messages. Each active user picks randomly a codeword $\cv_{\rv{w}_i}$ from a common codebook containing $M$ codewords $\cv_1,\dots,\cv_M$ drawn independently from the distribution $\Cc\Nc(\mathbf{0},P'\Id_n)$ for a fixed $P' < P$.  To convey message $\rv{w}_i$, user $i$ transmits $\cv_{\rv{w}_i}$ provided that $\|\cv_{\rv{w}_i}\|^2 \le nP$. Otherwise, it transmits the all-zero codeword.
The receiver employs a minimum distance decoder where the decoded list is $\widehat{\Wc} = \arg\min_{\widehat{\Wc} \subset [M], |\widehat{\Wc}| = K_{\rm a}} \|c(\widehat{\Wc}) - \rvVec{y}\|^2$, with $c(\Wc) \defeq \sum_{i\in \Wc} \cv_{i}$. The error analysis involves manipulations of unions of the pairwise error events via a change of measure and the application of the Chernoff bound combined with Gallager's $\rho$-trick~\cite[p.~136]{Gallager1968information}. 
An alternative bound is also obtained by writing the pairwise error event as an inequality involving information densities, and by applying a property of the information density given in~\cite[Cor.~17.1]{Polyanskiy2019lecture}.

In the following, we derive a similar random-coding bound for the case in which $\rv{K}_{\rm a}$ is random and unknown to the receiver. Specifically, we consider a random-coding scheme with the same encoder as in \cite{PolyanskiyISIT2017massive_random_access}. However, since the receiver does not know $\rv{K}_{\rm a}$, we let the decoder estimate $\rv{K}_{\rm a}$ from $\rvVec{y}$, then decide the best list size within an interval around the initially detected value of $\rv{K}_{\rm a}$. Specifically, given the channel output $\yv$, the receiver estimates $\rv{K}_{\rm a}$ as
\begin{align}
	K_{\rm a}' = \arg\min_{K \in [K_l:K_u]} m(\yv,K),
\end{align}
where $m(\yv,K)$ is a suitably chosen metric, and $K_l$ and $K_u$ are suitably chosen lower and upper limits on $K_{\rm a}'$, respectively. 
Then, given $K_{\rm a}'$, the receiver decodes the list of messages as
\begin{equation} \label{eq:decoder_Ka'}
	\widehat{\Wc} = \arg\min_{\widehat{\Wc} \subset [M], \underline{K_{\rm a}'} \le |\widehat{\Wc}| \le \overline{K_{\rm a}'}} \|c(\widehat{\Wc}) - \rvVec{y}\|^2,
\end{equation}   
		where \revisee{$\underline{K_{\rm a}'} =\max\{K_l,K_{\rm a}'-r\}$ and $\overline{K_{\rm a}'}\defeq \min\{K_u,K_{\rm a}'+r\}$ with a chosen nonnegative integer $r$}. 
	\revisee{We refer to $r$ as the {\em decoding radius}.}
An error analysis of this random-coding scheme conducted along similar lines as in \cite{PolyanskiyISIT2017massive_random_access} leads to the following result.

\begin{theorem}[Random-coding bound, $\rv{K}_{\rm a}$ random and unknown]  \label{thm:RCU_unknownKa}
	Fix $P' < P$, \revise{$r$, $K_{l}$, and $K_{u}$ ($K_{l} \le K_{u}$)}. For the $\rv{K}_{\rm a}$-user Gaussian \gls{MAC} with $\rv{K}_{\rm a} \sim P_{\rv{K}_{\rm a}}$, there exists an $(M,n,\epsilon_{\rm MD},\epsilon_{\rm FA})$ random-access code satisfying the power constraint $P$ and 
	\begin{align}
		\epsilon_{\rm MD} &= \sum_{K_{\rm a} =\max\{K_{l},1\}}^{K_{u}} \bigg(P_{\rv{K}_{\rm a}}(K_{\rm a}) \sum_{K_{\rm a}' = K_{l}}^{K_{u}} \sum_{t\in \Tc}\frac{t+(K_{\rm a}-\overline{K_{\rm a}'})^+}{K_{\rm a}} \notag \\
		&\qquad \cdot\min\{p_t,q_t, \xi(K_{\rm a},K_{\rm a}')\} \bigg) + p_0, \label{eq:eps_MD}\\
		\epsilon_{\rm FA} &= \sum_{K_{\rm a} =K_{l}}^{K_{u}} \bigg(P_{\rv{K}_{\rm a}}(K_{\rm a}) \sum_{K_{\rm a}' = K_{l}}^{K_{u}} \sum_{t\in \Tc} \sum_{t' \in \Tc_t}  \notag \\ 
		&\qquad \frac{t'+(\underline{K_{\rm a}'}-K_{\rm a})^+}{K_{\rm a} - t - {(K_{\rm a} - \overline{K_{\rm a}'})}^+ + t' + {(\underline{K_{\rm a}'}-K_{\rm a})}^+} \notag \\ 
		&\qquad \cdot \min\{p_{t,t'}, q_{t,t'}, \xi(K_{\rm a},K_{\rm a}')\} \bigg) + p_0, \label{eq:eps_FA}
	\end{align}	
	where 
	\begin{align}
		p_0 &= 2 - \sum_{K_{\rm a} = K_{l}}^{K_{u}}P_{\rv{K}_{\rm a}}(K_{\rm a}) - \E_{\rv{K}_{\rm a}}\left[\frac{M!}{M^{\rv{K}_{\rm a}}(M-\rv{K}_{\rm a})!} \right] \notag \\
		&\quad + \E[\rv{K}_{\rm a}]  \frac{\Gamma(n,nP/P')}{\Gamma(n)}, \label{eq:p0}\\
		p_t &= \sum_{t'\in \overline{\Tc}_t} p_{t,t'}, \label{eq:pt}\\
		p_{t,t'} &= e^{-n E(t,t')}, \label{eq:ptt} \\
		E(t,t') &= \max_{\rho,\rho_1 \in [0,1]} -\rho\rho_1 t' R_1 - \rho_1 R_2 + E_0(\rho,\rho_1), \label{eq:Ett} \\
		\!\!\!\! E_0(\rho,\rho_1) &= \max_{\lambda} \rho_1 a + \ln(1-\rho_1 P_2 b), \label{eq:E0}\\
		a &= \rho \ln(1+ P' t' \lambda) + \ln(1+ P't \mu), \label{eq:a}\\ 
		b &= \rho\lambda -\frac{\mu}{1+ P't\mu}, \label{eq:b} \\ 
		\mu &= \frac{\rho \lambda}{1+P't'\lambda}, \\
		P_2 &= 1+ \big((K_{\rm a} - \overline{K_{\rm a}'})^+ + (\underline{K_{\rm a}'} - K_{\rm a})^+\big)P', \label{eq:P2}\\
		R_1 &=  \frac{1}{nt'} \ln\binom{M - \max\{K_{\rm a},\underline{K_{\rm a}'}\}}{t'}, \label{eq:R1}
		\\
		R_2 &=  \frac{1}{n} \ln \binom{\min\{K_{\rm a}, \overline{K_{\rm a}'}\}}{t}, \\
		q_t &= \inf_{\gamma} \bigg(\!\P[\rv{I}_{t} \!\le\! \gamma] + \sum_{t'\in \overline{\Tc}_t}\!
		\exp(n(t'R_1 \!+\! R_2) \!-\! \gamma)\!\bigg), \label{eq:qt}\\
		q_{t,t'} &= \inf_{\gamma} \Big(\P[\rv{I}_{t} \!\le\! \gamma] + \exp(n(t'R_1 \!+\! R_2) \!-\! \gamma)\Big), \label{eq:qtt} \\
		\Tc &= [0:\min\{\overline{K_{\rm a}'},K_{\rm a},M\!-\!\underline{K_{\rm a}'} \!-\! (K_{\rm a} \!-\! \overline{K_{\rm a}'})^+\}], \label{eq:T} \\
		\Tc_t &= \big[\big({(K_{\rm a} - \overline{K_{\rm a}'})}^+ - {(\underline{K_{\rm a}'} - K_{\rm a})}^+ + \max\{\underline{K_{\rm a}'},1\} \big. \big. \notag \\
		&\quad \quad \big. \big. - K_{\rm a} + t\big)^+ : u_t\big],	\label{eq:Tt}	\\
		\overline{\Tc}_t &= \big[\big({(K_{\rm a} \!-\! \overline{K_{\rm a}'})}^+ - {(K_{\rm a}\!-\!\underline{K_{\rm a}'})}^+ + t\big)^+ : u_t \big], \label{eq:Tbart} \\
		u_t &= \min\big\{{(\overline{K_{\rm a}'} - K_{\rm a})}^+ - {(\underline{K_{\rm a}'}-K_{\rm a})}^+ + t,  \big.\notag \\
		&\quad \quad \big. \overline{K_{\rm a}'} - {(\underline{K_{\rm a}'}\!-\!K_{\rm a})}^+, M-\max\{\underline{K_{\rm a}'},K_{\rm a}\}\big\}, \\
		\!\!\!\!\!\!\xi(K_{\rm a},K_{\rm a}') &= 
			\min_{K\colon K \ne K_{\rm a}'} \P[m\left(\rvVec{y}_0,K_{\rm a}' \right) < m\left(\rvVec{y}_0,K\right)], \label{eq:xi}
	\end{align}
	in~\eqref{eq:xi}, $\rvVec{y}_0 \sim \Cc\Nc(\mathbf{0},(1+K_{\rm a}P')\Id_n)$. The random variable $\rv{I}_t$ in~\eqref{eq:qt} and \eqref{eq:qtt} is defined as
	\begin{equation} \label{eq:def_It}
		\rv{I}_t \defeq \!\!\min_{\Wc_{02} \subset [(K_{\rm a} - \overline{K_{\rm a}'})^+ + 1:K_{\rm a}] \atop |\Wc_{02}| = t}\!\! \imath_t(c(\Wc_{01}') + c(\Wc_{02});\rvVec{y} \cond c([K_{\rm a}] \setminus \Wc_0)),
	\end{equation} 
	where $\Wc_{01}' = [K_{\rm a} + 1: \underline{K_{\rm a}'}]$, $\Wc_0 = [(K_{\rm a} - \overline{K_{\rm a}'})^+] \cup \Wc_{02}$, 
	and 
	\begin{align}
		&\imath_t(c(\Wc_{0});\rvVec{y} \cond c(\Wc \setminus \Wc_0)) \notag \\
		&= n \ln(1+(t+(K_a\!\!-\overline{K_{\rm a}'})^+)P') + \frac{\|\rvVec{y} - c(\Wc \setminus \Wc_0)\|^2}{1+(t+(K_a\!-\!\overline{K_{\rm a}'})^+)P'} \notag \\
		&\quad - \|\rvVec{y} - c(\Wc_0) - c(\Wc \setminus \Wc_0)\|^2.  \label{eq:infor_den}
	\end{align}
\end{theorem}

Some remarks are in order.
\begin{enumerate}[leftmargin=*,label={\roman*)}]
	\item The parameters $K_{l}$ and $K_{u}$ can be taken to be the essential infimum and the essential supremum of $\rv{K}_{\rm a}$, respectively. In numerical evaluation, it is often convenient to set $K_{l}$ to be the largest value and $K_{u}$ the smallest value for which $\sum_{K_{\rm a} = K_{l}}^{K_{u}}P_{\rv{K}_{\rm a}}(K_{\rm a})$ exceeds a predetermined threshold.  
	
	\item The term $1 - \E_{\rv{K}_{\rm a}}\left[\frac{M!}{M^{\rv{K}_{\rm a}}(M-\rv{K}_{\rm a})!} \right]$ in $p_0$ can be upper-bounded by $\E_{\rv{K}_{\rm a}}\big[\binom{\rv{K}_{\rm a}}{2}/M\big]$ as in \cite{PolyanskiyISIT2017massive_random_access}. 
	
	\item The term $R_1$ in~\eqref{eq:R1} can be upper-bounded by $ \frac{1}{n} \ln (M - \max\{K_{\rm a},\underline{K_{\rm a}'}\}) - \frac{1}{nt'} \ln t'!$, which allows for a stable computation when $M - \max\{K_{\rm a},\underline{K_{\rm a}'}\}$ is large.
	
	\item The optimal $\lambda$ in~\eqref{eq:E0} is given by the largest real root of the cubic function $c_1x^3 + c_2x^2 + c_3x + c_4$ with
	\begin{align}
		c_1 &=  -\rho \rho_1(\rho\rho_1 + 1)t'P'P_2P_3^2,\\
		c_2 &= \rho\rho_1 t'P'P_3^2 - \rho\rho_1(3-\rho_1)t'P'P_2P_3 \notag \\
		&\quad -\rho\rho_1(\rho_1+1)P_2P_3^2,\\
		c_3 &= (2\rho-1)\rho_1 t'P'P_3 + \rho_1 P_3^2 - 2\rho\rho_1 P_2P_3, \\
		c_4 &= (\rho-1)\rho_1 t'P' + \rho_1P_3, 
	\end{align}
	where $P_2$ is given by~\eqref{eq:P2} and $P_3 \defeq (t' + \rho t)P'$.
	
	\item \revise{If the number of active users is fixed to $K_{\rm a}$, by letting $K_{\rm a}' = K_{\rm a}$ with probability $1$ and \revisee{by} setting the decoding radius $r$ to $0$, 
		one obtains from Theorem~\ref{thm:RCU_unknownKa} a trivial generalization of \cite[Th.~1]{PolyanskiyISIT2017massive_random_access} to the complex case.}
\end{enumerate}

\begin{proof}[Proof of Theorem~\ref{thm:RCU_unknownKa}]
 	We next present a sketch of the proof. The full proof can be found in Appendix~\ref{app:proof}. 
	
	Denote by $\Wc_0$ the set of misdetected messages, i.e., $\Wc_0 \defeq \Wc \setminus \widehat{\Wc}$, and by $\Wc_0'$ the set of falsely alarmed messages, i.e., $\Wc_0' \defeq \widehat{\Wc} \setminus \Wc$. The \gls{MD} and \gls{FA} probabilities, given in~\eqref{eq:eps_MD} and \eqref{eq:eps_FA}, respectively, can be expressed as 
	$P_{\rm MD} = \E[\ind{|\Wc| \ne 0} \cdot \md]$ and $P_{\rm FA} = \E[\ind{|\widehat{\Wc}| \ne 0} \cdot \fa]$.
	At a cost of adding a constant bounded by $p_0$ given in~\eqref{eq:p0}, we first replace the measure over which the expectation is taken by the one under which: i) there are at least $K_{l}$ and at most $K_{u}$ active users; ii) $\widetilde{\rv{w}}_1,\dots,\widetilde{\rv{w}}_{\rv{K}_{\rm a}}$ are sampled uniformly without replacement from $[M]$; iii) $\rvVec{x}_i = \cv_{\rv{w}_i} \forall i,$ instead of $\rvVec{x}_i = \cv_{\rv{w}_i} \ind{\|\cv_{\rv{w}_i}\|^2 \le nP}$.
	
	Let $K_{\rm a} \to K_{\rm a}'$ denote the event that the estimation step outputs $K_{\rm a}'$ while $K_{\rm a}$ users are active.
	Given $K_{\rm a} \to K_{\rm a}'$, note that if $\overline{K_{\rm a}'} < K_{\rm a}$, the decoder commits at least $K_{\rm a} - \overline{K_{\rm a}'}$ \glspl{MD}; if $\underline{K_{\rm a}'} > K_{\rm a}$, the decoder commits at least $\underline{K_{\rm a}'} - K_{\rm a}$ \glspl{FA}. We let $\Wc_0 = \Wc_{01} \cup \Wc_{02}$ where $\Wc_{01}$ denotes the list of $(\rv{K}_{\rm a} - \overline{K_{\rm a}'})^+$ \textit{initial} \glspl{MD} due to insufficient decoded list size, and $\Wc_{02}$ the \textit{additional} \glspl{MD} occurring during decoding. Similarly, let $\Wc_0' = \Wc_{01}' \cup \Wc_{02}'$ where $\Wc_{01}'$ denotes the list of $(\underline{K_{\rm a}'}-\rv{K}_{\rm a})^+$ \textit{initial} \glspl{FA} due to excessive decoded list size, and $\Wc_{02}'$ the \textit{additional} \glspl{FA}. Fig.~\ref{fig:venn} depicts the
	relation between these sets.  
	\begin{figure}
		\centering
		\begin{tikzpicture}[thick,scale=0.95, every node/.style={scale=0.95}]
			\def\radius{2cm}
			\def\radiusB{0.9*\radius}
			\def\mycolorbox#1{\textcolor{#1}{\rule{2ex}{2ex}}}
			\colorlet{colori}{gray!80}
			\colorlet{colorii}{gray!20}
			
			\coordinate (ceni) at (0,0);
			\coordinate[xshift=1.05*\radius] (cenii);
			
			\coordinate (edge1a) at (-\radiusB,0.1cm);
			\coordinate(edge1b) at (\radiusB,-.2cm);
			
			\coordinate (edge2a) at (\radius-\radiusB-.1cm,.1cm);
			\coordinate (edge2b) at (\radius+\radiusB+.3cm,-.2cm);
			
			\draw[fill=colori,fill opacity=0.5] (ceni) circle (\radiusB);
			\draw[fill=colorii,fill opacity=0.5] (cenii) circle (\radius);
			\draw (ceni) circle (\radiusB);
			
			\draw (edge1a) to (edge2a);
			\draw (edge1b) to (edge2b);
			
			\draw[-latex] (-\radius,-0.6*\radius) node[below,xshift=-.4cm,text width=2cm,align=center] {\small Transmitted messages $\Wc$} -- (-0.77*\radius,-0.5*\radius);
			\draw[-latex] (2.13*\radius,-0.6*\radius) node[below,xshift=.3cm,text width=2cm,align=center] {\small Decoded messages $\widehat{\Wc}$} -- (1.93*\radius,-0.5*\radius);
			
			\node[yshift=1.1*\radius,xshift=-1.25cm,text width=4cm,align=center] {\small \glspl{MD}: \\ $\Wc_0 = \Wc_{01} \cup \Wc_{02} = \Wc \setminus \widehat{\Wc}$};
			
			\node[yshift=1.2*\radius,xshift=1.25cm,text width=4cm,align=center] at (cenii) {\small \glspl{FA}: \\ $\Wc_0' = \Wc'_{01} \cup \Wc'_{02} = \widehat{\Wc} \setminus \Wc$};
			
			\node[xshift=.93cm,text width=1.2cm,align=center] at (ceni) {\small correctly decoded messages $\Wc \cap \widehat{\Wc}$};
			\node[yshift=.8\radius,xshift=-.62cm,text width=2.1cm,align=center] at (ceni) {\small $~~~~(\rv{K}_{\rm a}-\overline{K_{\rm a}'})^+$ initial \glspl{MD} $~~~\Wc_{01}~~~$};
			\node[yshift=-.8\radius,xshift=-.7\radius,text width=1.5cm,align=center] at (ceni) {\small additional \glspl{MD} $~~~\Wc_{02}~~~$};
			\node[yshift=.7\radius,xshift=.6\radius,text width=1.5cm,align=center] at (cenii) {\small $(\underline{K_{\rm a}'}-\rv{K}_{\rm a})^+$ initial \glspl{FA} $~~~\Wc'_{01}~~~$};
			\node[yshift=-1cm,xshift=.7\radius,text width=1.5cm,align=center] at (cenii) {\small additional \glspl{FA} $~~~\Wc'_{02}~~~$};
		\end{tikzpicture}
		\caption{A diagram depicting the relation between the defined sets of messages.}
		\label{fig:venn}
	\end{figure}
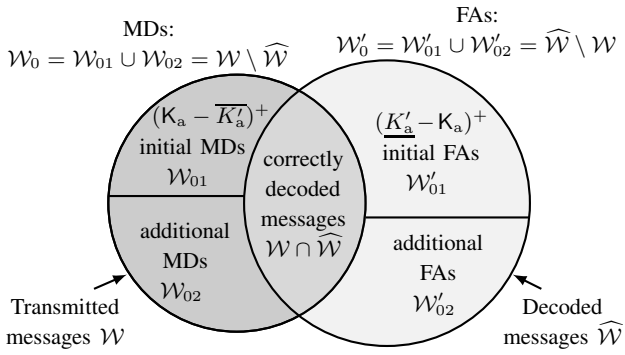

Using the above definitions, the set of transmitted messages is $\Wc = \Wc_{01} \cup \Wc_{02} \cup (\Wc \setminus \Wc_0)$, and the received signal is $\rvVec{y} = c(\Wc_{01}) + c(\Wc_{02}) + c(\Wc \setminus \Wc_0) + \rvVec{z}$. Since the messages in $\Wc_{01}$ are always misdetected and the messages in $\Wc_{01}'$ are always falsely alarmed, the best approximation of $\Wc$ that the decoder can produce is $\Wc_{02} \cup (\Wc \setminus \Wc_0) \cup \Wc_{01}'$. However, under the considered error event $\Wc \to \widehat{\Wc}$, the actual decoded list is $\Wc'_{02} \cup (\Wc \setminus \Wc_0) \cup \Wc_{01}'$. Therefore, $\Wc \to \widehat{\Wc}$ implies the event $F(\Wc_{01},\Wc_{02},\Wc_{01}',\Wc_{02}') \defeq \big\{\|c(\Wc_{01}) + c(\Wc_{02})- c(\Wc_{01}') - c(\Wc_{02}') + \rvVec{z}\|^2 < \|c(\Wc_{01}) - c(\Wc_{01}') + \rvVec{z}\|^2 \big\}$.

It follows that, after the change of measure, $P_{\rm MD}$ and $P_{\rm FA}$ can be bounded as
\begin{align}
	P_{\rm MD} &\le \sum_{K_{\rm a} =\max\{K_{l},1\}}^{K_{u}} \!\!\bigg(P_{\rv{K}_{\rm a}}(K_{\rm a}) \sum_{K_{\rm a}' = K_{l}}^{K_{u}} \sum_{t\in \Tc}\frac{t+(K_{\rm a}-\overline{K_{\rm a}'})^+}{K_{\rm a}} \notag \\
	&\qquad \cdot \P[|\Wc_{02}| = t, K_{\rm a} \to K_{\rm a}'] \bigg) + p_0,  \\
	P_{\rm FA} &\le \sum_{K_{\rm a} =K_{l}}^{K_{u}} \bigg(P_{\rv{K}_{\rm a}}(K_{\rm a}) \sum_{K_{\rm a}' = K_{l}}^{K_{u}}  \sum_{t\in \Tc} \sum_{t' \in \Tc_t} \notag \\
	&\qquad \frac{t+(\underline{K_{\rm a}'} - K_{\rm a})^+}{K_{\rm a} \!-\! t \!-\! {(K_{\rm a} \!-\! \overline{K_{\rm a}'})}^+ \!\!+\! t' \!+\! {(\underline{K_{\rm a}'}\!-\!K_{\rm a})}^+\!} \notag \\
	&\qquad \cdot \P[|\Wc_{02}| = t, |\Wc_{02}'| = t', K_{\rm a} \to K_{\rm a}'] \bigg) \!+\! p_0, \label{eq:tmp853}
\end{align} 
where $\Tc$ and $\Tc_t$ are given by~\eqref{eq:T} and \eqref{eq:Tt}, respectively. The constraint $t \in \Tc$ holds because the number of \glspl{MD}, given by $t+{(K_{\rm a}-\overline{K_{\rm a}'})}^+$, is upper-bounded by the total number of transmitted messages $K_{\rm a}$, and by $M-\underline{K_{\rm a}'}$ (since at least $\underline{K_{\rm a}'}$ messages are returned). The constraint $t' \in \Tc_t$ holds because: i) the decoded list size, given by $K_{\rm a} - t - {(K_{\rm a} - \overline{K_{\rm a}'})}^+ + t' + {(\underline{K_{\rm a}'}-K_{\rm a})}^+$, must be in $[\underline{K_{\rm a}'} : \overline{K_{\rm a}'}]$ and must be positive since the event $|\widehat{\Wc}| = 0$ results in no \gls{FA} by definition; ii) the number of \glspl{FA}, given by $t'+ (\underline{K_{\rm a}'}-K_{\rm a})^+$, is upper-bounded by the number of messages that are not transmitted $M-K_{\rm a}$, and by the maximal number of decoded messages $\overline{K_{\rm a}'}$.

\revisee{Let $A(K_{\rm a},K_{\rm a}') \defeq
	\{\rvVec{y} \colon K_{\rm a}' = \arg\min_{K \in [K_l : K_u]} m(\rvVec{y},K) \}$.}
Since the event $K_{\rm a} \to K_{\rm a}'$ implies $|\widehat{\Wc}| \in [\underline{K_{\rm a}'}:\overline{K_{\rm a}'}]$ and $A(K_{\rm a},K_{\rm a}')$, we have that 
\begin{align*}
	&\P[|\Wc_{02}| = t, K_{\rm a} \to K_{\rm a}'] \notag \\
	&\le \P[{|\Wc_{02}| = t, |\widehat{\Wc}|\in [\underline{K_{\rm a}'}:\overline{K_{\rm a}'}],A(K_{\rm a},K_{\rm a}')}] \\
	&\le \min\big\{\P\big({|\Wc_{02}| \!=\! t, |\widehat{\Wc}|\in [\underline{K_{\rm a}'}:\overline{K_{\rm a}'}]}\big),  \P[A(K_{\rm a},K_{\rm a}')] \big\}.
\end{align*}
Similarly, $\P[{|\Wc_{02}| = t,|\Wc_{02}'| = t', K_{\rm a} \to K_{\rm a}'}] \le \min\Big\{\P\Big[|\Wc_{02}| \!=\! t, |\Wc_{02}'| \!=\! t', |\widehat{\Wc}|\!\in\! [\underline{K_{\rm a}'}:\overline{K_{\rm a}'}]\Big]$, $\P[A(K_{\rm a},K_{\rm a}')] \Big\}.$

Under the new measure, $\rvVec{y} \sim \Cc\Nc(\mathbf{0},(1+K_{\rm a} P')\Id_n)$. Thus, we can show that $\P[A(K_{\rm a},K_{\rm a}')]$ is upper-bounded by $\xi(K_{\rm a},K_{\rm a}')$ given by~\eqref{eq:xi}. To establish \eqref{eq:eps_MD} and \eqref{eq:eps_FA}, we proceed as in \cite{PolyanskiyISIT2017massive_random_access}, write the events $\{|\Wc_{02}| \!=\! t, |\widehat{\Wc}| \!\in\! [\underline{K_{\rm a}'}:\overline{K_{\rm a}'}]\}$ and $\{|\Wc_{02}| \!=\! t,|\Wc_{02}'| \!=\! t', |\widehat{\Wc}|\!\in\! [\underline{K_{\rm a}'}:\overline{K_{\rm a}'}]\}$ as the union of 
	$F(\Wc_{01},\Wc_{02},\Wc_{01}',\Wc_{02}')$ events, and bound their probabilities by $\min\{p_t, q_t\}$ and $\min\{p_{t,t'}, q_{t,t'}\}$, respectively. 	
	
	\revisee{Finally, to guarantee that {\em both} \eqref{eq:eps_MD} and \eqref{eq:eps_FA} are satisfied, we allow for randomized coding strategies, by introducing the variable $\rv{U}$, which acts as common randomness. Proceeding as in \cite[Th.~19]{Polyanskiy2011feedback}, one can show that
		it is sufficient to perform randomization across (at most) three deterministic codes, i.e., $|\Uc|\le 3$.}
\end{proof}

	In the following proposition, we derive $\xi(K_{\rm a},K_{\rm a}')$ for two different estimators of $\rv{K}_{\rm a}$.
	\begin{proposition} \label{prop:xi}
		For the \gls{ML} estimation of $\rv{K}_{\rm a}$, i.e., $m(\yv,K) \!=\! -\ln p_{\rvVec{y} | \rv{K}_{\rm a}}(\yv | K)$, $\xi(K_{\rm a}, K_{\rm a}')$ is given by
		\begin{align} \label{eq:xi_ML}
			\xi(K_{\rm a},K_{\rm a}') &\defeq \min_{K:\; K \ne K_{\rm a}'} \!\Big(\ind{K \!<\! K_{\rm a}'}\frac{\Gamma(n,\zeta(K,K_{\rm a},K_{\rm a}')}{\Gamma(n)} \notag \\ 
			&	\qquad ~~ + \ind{K \!>\! K_{\rm a}'}\frac{\gamma(n,\zeta(K,K_{\rm a},K_{\rm a}'))}{\Gamma(n)}\Big),
		\end{align}
		with 
		\begin{align}
			\zeta(K,K_{\rm a},K_{\rm a}') &\defeq n \ln\left(\frac{1+KP'}{1+K_{\rm a}'P'}\right)(1+K_{\rm a}P')^{-1} \notag \\
			&\quad \cdot \left(\frac{1}{1+K_{\rm a}'P'}-\frac{1}{1+KP'}\right)^{-1}. \label{eq:zeta_ML}
		\end{align}
	
		For an energy-based estimation of $\rv{K}_{\rm a}$ with $m(\yv,K) = |\|\yv\|^2 - n(1 + KP')|$, $\xi(K_{\rm a},K_{\rm a}')$ is given by~\eqref{eq:xi_ML} with \begin{align}
			\zeta(K,K_{\rm a},K_{\rm a}') \defeq \frac{n}{1+K_{\rm a}P'}\left(1+\frac{K + K_{\rm a}'}{2}P'\right). \label{eq:zeta_energy}
		\end{align}
		
	\end{proposition}
	\begin{proof}
		See Appendix~\ref{proof:xi}.
	\end{proof}

\revisee{The decoding radius $r$ can be optimized according to the target \gls{MD} and \gls{FA} probabilities. A large decoding radius reduces the initial \glspl{MD} and \glspl{FA} at the cost of overfitting, especially at low SNR values. Specifically, when the noise dominates, increasing $r$ increases the chance that the decoder~\eqref{eq:decoder_Ka'} returns a list whose sum is closer in Euclidean distance to the noise than to the sum of the transmitted codewords.}
	

Our random-coding bound can also be applied to \revisee{SA-MPR} to investigate the resulting \gls{MD}--\gls{FA} trade-off. Consider an \revisee{SA-MPR} scheme where a length-$n$ frame is divided into $L$ slots and each user chooses randomly a slot to transmit. For $K_{\rm a}$ active users, the number of users transmitting in a slot follows a Binomial distribution with parameter $(K_{\rm a},1/L)$. 
The \gls{PMF} of the number of active users per slot, denoted by $\rv{K}_{\rm {SA}}$, is given by
	$P_{\rv{K}_{\rm {SA}}}(K_{\rm SA}) = \sum_{K_{\rm a}} P_{\rv{K}_{\rm a}}(K_{\rm a}) \binom{K_{\rm a}}{K_{\rm SA}} L^{-K_{\rm SA}}\left(1-\frac{1}{L}\right)^{K_{\rm a} - K_{\rm SA}}$.
Existing analyses of slotted ALOHA usually assume that the decoder can detect perfectly if no signal, one signal, or more than one signal have been transmitted in a slot. Furthermore, it is usually assumed that a collision-free slot leads to successful message decoding. However, the more slots, the shorter the slot length over which a user transmits its signal. 
To account for both detection and decoding errors, in Corollary~\ref{coro:slotted_ALOHA} below, we apply our decoder in a slot-by-slot manner, and obtain a random-coding bound similar to Theorem~\ref{thm:RCU_unknownKa}. 

\begin{corollary} \label{coro:slotted_ALOHA}
	For the Gaussian \gls{MAC} with the number of active users following $P_{\rv{K}_{\rm a}}$ and frame length $n$, an SA-MPR scheme with $L$ slots can achieve the \gls{MD} and \gls{FA} probabilities given in~\eqref{eq:eps_MD} and \eqref{eq:eps_FA}, respectively, with codebook size $M$, codeword length $n/L$, power constraint $PL$, and per-slot number of active users following $P_{\rv{K}_{\rm {SA}}}$.
\end{corollary}


\section{Numerical Evaluation} \label{sec:numerical}
In this section, we numerically evaluate the proposed random-coding bound and compare it with \revisee{the} random-access \revisee{bound/}schemes \revisee{in\cite{PolyanskiyISIT2017massive_random_access,Fengler2019sparcs,Amalladinne2020unsourced}}. We assume that $\rv{K}_{\rm a}$ follows a Poisson distribution. 
We consider $k \defeq \log_2 M = 128$ bits and $n = 19200$ complex channel uses (i.e., $38400$ real degrees of freedom). The power is given in terms of average bit energy $\EbNo \defeq nP/k$.

\begin{figure}[t!]
	\centering
\definecolor{magenta}{rgb}{1.00000,0.00000,1.00000}%
\begin{tikzpicture}
	\tikzstyle{every node}=[font=\scriptsize]
	\begin{axis}[%
		width=3in,
		height=2in,
		at={(0.759in,0.481in)},
		scale only axis,
		xmin=0,
		xmax=300,
		xlabel style={font=\color{black},yshift=1ex},
		xlabel={\footnotesize Average number of active users $\E[\rv{K}_{\rm a}]$},
		ymin=-2,
		ymax=12,
		ytick={-2, 0, 2, 4, 6, 8, 10, 12},
		ylabel style={font=\color{black}, yshift=-3.2ex},
		ylabel={\footnotesize Required $\EbNo$ (dB)},
		axis background/.style={fill=white},
		xmajorgrids,
		ymajorgrids,
		legend style={at={(0.05,-0.18)}, anchor=north west, row sep=-2.5pt, legend cell align=left, align=left, draw=white!15!black}
		]
		\addplot [color=green,line width=1pt]
		table[row sep=crcr]{%
			17 .7030 \\
			20 0.6092 \\
			25 0.6042 \\
			50 0.6120 \\ 
			75 0.6370 \\ 
			100 0.6701 \\ 
			125 0.7329 \\
			150 0.9365 \\
			175 1.2155 \\ 
			200 1.5259 \\ 
			225 1.8565 \\ 
			250 2.2029 \\ 
			275 2.5685 \\
			300 2.9512 \\
		};
		
		\addplot [color=blue, dashed,line width=1pt]
		table[row sep=crcr]{%
			2	-0.2898 \\
			10	-0.0988 \\
			25	0.0146 \\ 
			50	0.1275 \\ 
			75	0.1615 \\ 
			100	0.1974 \\ 
			125	0.3502 \\ 
			150	0.5912 \\ 
			175	0.8533 \\ 
			200	1.1264 \\ 
			225	1.4078 \\ 
			250	1.6960 \\ 
			275	1.9234 \\ 
			300	2.2876 \\
		};
		
		\addplot [color=magenta, line width=.8pt, mark=*, mark size=1.5pt, mark options={solid,color = magenta,fill=white}]
		table[row sep=crcr]{%
			5 -0.5104 \\ 
			10 -0.3337\\ 
			25 0.5329\\ 
			50 1.5650\\ 
			75 2.8320\\ 
			100 3.7426\\ 
			125 4.7103\\ 
			150 5.8035\\ 
			175 6.8849\\ 
			200 8.0970\\ 
			225 8.9092\\ 
			250 10.1469\\ 
			275 11.3445\\ 
			300 13.0720\\ 
		};
		
		\addplot [color=magenta, line width=.8pt, mark size=1.5pt, mark=square*, mark options={solid,color = magenta,fill=white},dashed]
		table[row sep=crcr]{%
			10 -.4161\\ 
			25 .47 \\
			50 1.2404 \\
			75 2.3500 \\
			100 3.3268 \\
			125 4.4576 \\
			150 5.4664 \\
			175 6.5878 \\
			200 7.7616 \\
			225 8.4016 \\
			250 9.5047 \\
			275 10.8719 \\
			300 12.1196 \\
		};
		
		\addplot [color=magenta, line width=.8pt, mark size=1.5pt, mark=*, mark options={solid, magenta}]
		table[row sep=crcr]{%
			5 -0.6204 \\ 
			10 -0.5139 \\ 
			25 0.3957 \\ 
			50 1.2324 \\ 
			75 2.1412 \\ 
			100 3.2419\\ 
			125 4.2234 \\ 
			150 5.2535 \\ 
			175 6.2287\\ 
			200 7.2511\\ 
			225 8.1388\\ 
			250 9.1707\\ 
			275 10.2813\\ 
			300 11.8315 \\ 
		};
		
		\addplot [color=magenta, dashed, line width=.8pt, mark size=1.3pt, mark=square*, mark options={solid,color = magenta}]
		table[row sep=crcr]{%
			5 -.2949 \\ 
			10 -.4201 \\ 
			25 0.2775 \\ 
			50 1.0264 \\
			75 1.8417 \\
			100 2.8363 \\
			125 3.9383 \\ 
			150 4.8895 \\ 
			175 5.8534 \\ 
			200 6.7662 \\ 
			225 7.5140 \\ 
			250 8.4517 \\ 
			275 9.8268 \\
			300 10.8815 \\
		};

		\addplot [color=orange, mark=x, line width=.8pt, mark size=2pt, mark options={solid, orange}]
		table[row sep=crcr]{%
			15 1.889065225218457877e+00 \\
			20 1.829007583756270039e+00 \\
			50 2.152915984617741252e+00 \\
			75 2.35993197933708437e+00 \\
			100 2.5614794499550815 \\
			125 2.8702 \\
			150 3.2059 \\
			175 3.6287 \\
			200 4.2303 \\
			225 4.6432 \\
			250 5.1734 \\
			300 6.6150 \\
		};
	
		\addplot [color=orange, mark=+, line width=.8pt, mark size=2pt, mark options={solid, orange}]
		table[row sep=crcr]{%
			15 3.0828 \\
			25 3.0260 \\
			50 3.1830 \\
			75 3.2404 \\
			100 3.2987687520233147\\
			125 3.3769650702680227 \\
			150 3.5981 \\
			175 4.0657 \\
			200 4.5905 \\
			225 4.9590 \\
			250 5.5682 \\
			300 7.0733 \\
		};
	
		
		\draw [color=gray] (200,12) node [below,xshift=5pt,yshift=1pt,color=black] {\cite[Th.~1]{PolyanskiyISIT2017massive_random_access}} -- (194,30);
		
		\draw [color=gray] (260,12) node [below,xshift=5pt,yshift=1pt,color=black] {{Theorem~\ref{thm:RCU_unknownKa}}} -- (245,41);
		
		
		
		
		\draw [black] (275,120) ellipse [x radius=5, y radius=6];
		\draw [color=gray] (200,127) node [left,color=black,xshift=4pt,yshift=0pt,text width=1.9cm,align=center] {{SA-MPR with slot-index coding}} -- (272,124);
		
		\draw [black] (200,100.5) ellipse [x radius=5, y radius=6];
		\draw [color=gray] (100,115) node [left,color=black] {{SA-MPR}} -- (196,104);
		
		\draw [color=gray] (26,75) node [above,color=black,xshift=0pt,yshift=-2pt,text width=1.5cm,align=center] {SPARC~\cite{Fengler2019sparcs}} -- (32,51);
		
		\draw [color=gray] (120,12) node [below,color=black,xshift=5pt,yshift=1pt] {Enhanced SPARC~\cite{Amalladinne2020unsourced}} -- (112,47);
	\end{axis}
\end{tikzpicture}%
	\caption{The required $\EbNo$ to achieve $\max\{P_{\rm MD}, P_{\rm FA}\} \le 0.1$ vs. $\E[\rv{K}_{\rm a}]$ for $k=128$ bits, $n = 19200$ channel uses, and $\rv{K}_{\rm a} \sim \mathrm{Poisson}(\E[\rv{K}_{\rm a}])$. We compare our random coding bound (Theorem~\ref{thm:RCU_unknownKa}) and SA-MPR bound (Corollary~\ref{coro:slotted_ALOHA}) with the bound in~\cite[Th.~1]{PolyanskiyISIT2017massive_random_access} for $\rv{K}_{\rm a}$ known, and two practical schemes, namely, 
		the SPARC scheme~\cite{Fengler2019sparcs}, and its enhancement~\cite{Amalladinne2020unsourced}. 
		Solid lines represent schemes/bounds with $\rv{K}_{\rm a}$ unknown; dashed lines represent schemes/bounds with $\rv{K}_{\rm a}$ known.}
	\label{fig:EbN0_vs_EKa}
\end{figure}
In Fig.~\ref{fig:EbN0_vs_EKa}, we compare our random-coding bound with that of Polyanskiy~\cite{PolyanskiyISIT2017massive_random_access} in terms of the required $\EbNo$ so that neither $P_{\rm MD}$ nor $P_{\rm FA}$ exceed $0.1$. For our bound, we consider the \gls{ML} estimator of $\rv{K}_{\rm a}$ and zero decoding radius, i.e., $\overline{K_{\rm a}'} = \underline{K_{\rm a}'} = K_{\rm a}'$. Numerical evaluation suggests that this choice is optimal for these target MD and FA probabilities. \revise{We choose $K_l$ to be the largest value and $K_u$ the smallest value for which $\P[{\rv{K}_{\rm a} \notin [K_l,K_u]}] < 10^{-9}$. The $q_t$ and $q_{t,t'}$ terms are evaluated for $t = 1$ and $K_{\rm a} \le 50$ only.}  For the bound in \cite[Th.~1]{PolyanskiyISIT2017massive_random_access}, we average over the Poisson distribution of $\rv{K}_{\rm a}$. This corresponds to the scenario where $\rv{K}_{\rm a}$ is random but known. As can be seen, the extra required $\EbNo$ due to the lack of knowledge of $\rv{K}_{\rm a}$ is about $0.5$--$0.7$~dB. 

In Fig.~\ref{fig:EbN0_vs_EKa}, we also show the performance of the SA-MPR bound given in Corollary~\ref{coro:slotted_ALOHA}, where we 
optimize $L$ and the decoding radius for each $\E[\rv{K}_{\rm a}]$. We also consider the possibility to encode $\lfloor \log_2 L\rfloor$ extra bits for each user in the slot index, and assume perfect decoding of these bits. We refer to this scheme as SA-MPR with slot-index coding. We also evaluate the performance of two practical schemes, namely:
\begin{itemize}
	
	\item the SPARC scheme proposed in~\cite{Fengler2019sparcs}, which employs a concatenated coding framework with an inner approximate message passing~(AMP) decoder followed by an outer tree decoder.
	
	\item an enhancement of the SPARC scheme proposed in~\cite{Amalladinne2020unsourced}, which we refer to as enhanced SPARC. This scheme introduces belief propagation between the inner AMP decoder and the outer tree decoder in an iterative manner. 
\end{itemize}
Note that 
the SPARC and enhanced SPARC schemes were proposed for the Gaussian \gls{MAC} with \textit{known} number of active users. To adapt these schemes to the case where $\rv{K}_{\rm a}$ is unknown, we employ an energy-based estimation of $\rv{K}_{\rm a}$, then treat this estimate as the true $\rv{K}_{\rm a}$ in the decoding process. From Fig.~\ref{fig:EbN0_vs_EKa}, we see that SA-MPR, even with slot-index coding, becomes power inefficient as $\E[\rv{K}_{\rm a}]$ increases. 
The enhanced SPARC scheme achieves the closest performance to our bound for $\E[\rv{K}_{\rm a}] \ge 100$. It outperforms the original SPARC scheme by about $0.5$~dB for large $\E[\rv{K}_{\rm a}]$. 

In Fig.~\ref{fig:Pe_vs_EbN0}, we plot the bounds on the \gls{MD} and \gls{FA} probabilities in Theorem~\ref{thm:RCU_unknownKa} (with ML estimation of $\rv{K}_{\rm a}$) as a function of $\EbNo$ for different decoding radii. We observe that decoding with a small radius performs better in the low $\EbNo$ regime, where noise overfitting is the  bottleneck. Increasing the decoding radius improves the performance in the moderate and high $\EbNo$ regime, where setting $r=0$ results in a high error floor due to the initial \glspl{MD} and \glspl{FA}. The error floor can be characterized analytically (see Appendix~\ref{app:error_floor}).
\begin{figure}[t!]
	\centering
	\definecolor{magenta}{rgb}{1.00000,0.00000,1.00000}%
	\begin{tikzpicture}
		\tikzstyle{every node}=[font=\scriptsize]
		\begin{axis}[%
			width=2.95in,
			height=1.8in,
			at={(0.759in,0.481in)},
			scale only axis,
			xmin=0,
			xmax=20,
			xlabel style={font=\color{black},yshift=1ex},
			xlabel={\footnotesize $\EbNo$ (dB)},
			ymin=7e-7,
			ymax=1,
			ymode = log,
			yminorticks=true,
			ytick={1e0, 1e-1, 1e-2, 1e-3, 1e-4, 1e-5, 1e-6},
			ylabel style={font=\color{black}, yshift=-1.8ex},
			ylabel={\footnotesize MD and FA probabilities},
			axis background/.style={fill=white},
			xmajorgrids,
			ymajorgrids,
			legend style={at={(0.99,0.99)}, anchor=north east, row sep=-2.5pt, legend cell align=left, align=left, draw=white!15!black}
			]			
			\addplot [color=blue,line width=0.5pt]
			table[row sep=crcr]{%
                   0   1.000000000000000 \\
0.500000000000000   0.146777707087554 \\
1.000000000000000   0.038501093391244 \\
1.500000000000000   0.025850383046193\\
2.000000000000000   0.021772071891965\\
2.500000000000000   0.018617775903273\\
3.000000000000000   0.016009283961793\\
3.500000000000000   0.013843186229793\\
4.000000000000000   0.012038715623758\\
4.500000000000000   0.010530615259555\\
5.000000000000000   0.009266057894810\\
5.500000000000000   0.008202197353721\\
6.000000000000000   0.007304145833348\\
7.000000000000000   0.005891122825509\\
8.000000000000000   0.004865075446055\\
9.000000000000000   0.004110837384321\\
10.000000000000000   0.003550069416826\\
11.000000000000000   0.003128846290549\\
12.000000000000000   0.002809571621660\\
13.000000000000000   0.002565675254348\\
14.000000000000000   0.002377690090178\\
15.000000000000000   0.002232321446917\\
16.000000000000000   0.002119370070190\\
17.000000000000000   0.002031256355991\\
18.000000000000000   0.001962289459717\\
19.000000000000000   0.001908159498677\\
20.000000000000000   0.001865577583167\\
			};
		\addlegendentry{$\epsilon_{\rm MD}$}
		
			\addplot [color=blue,dashed,line width=0.5pt]
			table[row sep=crcr]{%
                   0   1.000000000000000 \\
0.500000000000000   0.148101484266278 \\
1.000000000000000   0.038379904912468 \\
1.500000000000000   0.025609324195977 \\
2.000000000000000   0.021620294401635 \\
2.500000000000000   0.018549309949211 \\
3.000000000000000   0.015998992729046 \\
3.500000000000000   0.013872663930583 \\
4.000000000000000   0.012094931632125 \\
4.500000000000000   0.010604389871631 \\
5.000000000000000   0.009350962273943 \\
5.500000000000000   0.008293764145908 \\
6.000000000000000   0.007399234886732 \\
7.000000000000000   0.005987926360951 \\
8.000000000000000   0.004959840714961 \\
9.000000000000000   0.004202242566935 \\
10.000000000000000   0.003637935254935 \\
11.000000000000000   0.003213481211418 \\
12.000000000000000   0.002891445308809 \\
13.000000000000000   0.002645266730529 \\
14.000000000000000   0.002455441589217 \\
15.000000000000000   0.002308576012215 \\
16.000000000000000   0.002194423499029 \\
17.000000000000000   0.002105348035991 \\
18.000000000000000   0.002035612027232 \\
19.000000000000000   0.001980867561597 \\
20.000000000000000   0.001937795049910 \\
			};
		\addlegendentry{$\epsilon_{\rm FA}$}
		
		\addplot [color=black,dashdotted,line width=0.8pt]
		table[row sep=crcr]{%
			0   0.001704144946427 \\
			20 0.001704144946427 \\
		};
		\addlegendentry{Error floor}
	
		\addplot [color=black,dashdotted,line width=0.8pt]
		table[row sep=crcr]{%
			0  0.001774442973898 \\
			20 0.001774442973898 \\
		};
		
		\addplot [color=blue,line width=0.5pt]
		table[row sep=crcr]{%
                   0   1.000000000000000\\
0.500000000000000   1.000000000000000\\
1.000000000000000   1.000000000000000\\
1.500000000000000   1.000000000000000\\
2.000000000000000   0.079465507503007\\
2.500000000000000   0.033068587453746\\
3.000000000000000   0.007071139952214\\
3.500000000000000   0.004781463371322\\
4.000000000000000   0.003772564460892\\
4.500000000000000   0.003022039388262\\
5.000000000000000   0.002495748575896\\
5.500000000000000   0.001949281367915\\
6.000000000000000   0.001576417581834\\
7.000000000000000   0.001047338502548\\
8.000000000000000   0.000713831162347\\
9.000000000000000   0.000501282020659\\
10.000000000000000   0.000364090279862\\
11.000000000000000   0.000274299830336\\
12.000000000000000   0.000214431261303\\
13.000000000000000   0.000173702742435\\
14.000000000000000   0.000145309250922\\
15.000000000000000   0.000125189787585\\
16.000000000000000   0.000110685421613\\
17.000000000000000   0.000100057695447\\
18.000000000000000   0.000092159424436\\
19.000000000000000   0.000086218208996\\
20.000000000000000   0.000081703460546\\
		};
	
			\addplot [color=blue,dashed,line width=0.5pt]
	table[row sep=crcr]{%
                   0   1.000000000000000\\
0.500000000000000   1.000000000000000\\
1.000000000000000   1.000000000000000\\
1.500000000000000   1.000000000000000\\
2.000000000000000   0.388947507083411\\
2.500000000000000   0.220434978322149\\
3.000000000000000   0.111904745603824\\
3.500000000000000   0.096740896056475\\
4.000000000000000   0.089566128141456\\
4.500000000000000   0.083465134416361\\
5.000000000000000   0.003722992569717\\
5.500000000000000   0.002101844797902\\
6.000000000000000   0.001709628815437\\
7.000000000000000   0.001153906415675\\
8.000000000000000   0.000799477796390\\
9.000000000000000   0.000570967871549\\
10.000000000000000   0.000421525451198\\
11.000000000000000   0.000322480505532\\
12.000000000000000   0.000255637149437\\
13.000000000000000   0.000209649451649\\
14.000000000000000   0.000177328451787\\
15.000000000000000   0.000154178705677\\
16.000000000000000   0.000137362385795\\
17.000000000000000   0.000124961411813\\
18.000000000000000   0.000115696043535\\
19.000000000000000   0.000108695811219\\
20.000000000000000   0.000103357198913\\
	};

		\addplot [color=black,dashdotted,line width=0.8pt]
table[row sep=crcr]{%
	0  6.596366572455718e-05 \\
	20 6.596366572455718e-05 \\
};

\addplot [color=black,dashdotted,line width=0.8pt]
table[row sep=crcr]{%
	0 8.468424489119134e-05 \\
	20 8.468424489119134e-05 \\
};

		\addplot [color=blue,line width=0.5pt]
table[row sep=crcr]{%
                   0   1.000000000000000\\
0.500000000000000   1.000000000000000\\
1.000000000000000   1.000000000000000\\
1.500000000000000   1.000000000000000\\
2.000000000000000   0.371347651566337\\
2.500000000000000   0.104903277919702\\
3.000000000000000   0.026824695964395\\
3.500000000000000   0.003293703095084\\
4.000000000000000   0.000989809288806\\
4.500000000000000   0.000726657389976\\
5.000000000000000   0.000467791041960\\
5.500000000000000   0.000316023082713\\
6.000000000000000   0.000222432120851\\
7.000000000000000   0.000111613752367\\
8.000000000000000   0.000057888750466\\
9.000000000000000   0.000031379359686\\
10.000000000000000   0.000018024561585\\
11.000000000000000   0.000011056340908\\
12.000000000000000   0.000007251475433\\
13.000000000000000   0.000005059480403\\
14.000000000000000   0.000003745317576\\
15.000000000000000   0.000002916710074\\
16.000000000000000   0.000002375337933\\
17.000000000000000   0.000002009188516\\
18.000000000000000   0.000001754146781\\
19.000000000000000   0.000001572056363\\
20.000000000000000   0.000001438968208\\
};

\addplot [color=blue,dashed,line width=0.5pt]
table[row sep=crcr]{%
                   0   1.000000000000000\\
0.500000000000000   1.000000000000000\\
1.000000000000000   1.000000000000000\\
1.500000000000000   1.000000000000000\\
2.000000000000000   1.000000000000000\\
2.500000000000000   0.879722959591974\\
3.000000000000000   0.462610389296507\\
3.500000000000000   0.248376287959890\\
4.000000000000000   0.221089032429527\\
4.500000000000000   0.208364737070244\\
5.000000000000000   0.001001165755560\\
5.500000000000000   0.000383865642598\\
6.000000000000000   0.000275367786392\\
7.000000000000000   0.000144538657683\\
8.000000000000000   0.000078493487920\\
9.000000000000000   0.000044624176381\\
10.000000000000000   0.000026765448010\\
11.000000000000000   0.000017067693361\\
12.000000000000000   0.000011577483590\\
13.000000000000000   0.000008325764519\\
14.000000000000000   0.000006317145206\\
15.000000000000000   0.000005020374956\\
16.000000000000000   0.000004156985782\\
17.000000000000000   0.000003563938514\\
18.000000000000000   0.000003145616116\\
19.000000000000000   0.000002843890150\\
20.000000000000000   0.000002622481408\\
};

		\addplot [color=black,dashdotted,line width=0.8pt]
table[row sep=crcr]{%
	0  1.112629875990382e-06 \\
	20 1.112629875990382e-06 \\
};

\addplot [color=black,dashdotted,line width=0.8pt]
table[row sep=crcr]{%
	0 1.997909370360383e-06 \\
	20 1.997909370360383e-06 \\
};

		\draw [black] (191,-9.3) ellipse [x radius=2, y radius=.7];
		\draw [-latex] (180,-7.4) node [left] {Decoding radius $r = 1$} -- (190,-8.6);
		
		\draw [-latex] (180,-4.5) node [left] {Decoding radius $r = 0$} -- (190,-6.2);

		\draw [black] (191.5,-13.2) ellipse [x radius=2.5, y radius=.85];
		\draw [-latex] (180,-10.8) node [left] {Decoding radius $r = 2$} -- (191,-12.3);
		
		\end{axis}
	\end{tikzpicture}%
	\caption{The bounds on the \gls{MD} and \gls{FA} probabilities vs. $\EbNo$ for $k=128$ bits, $n = 19200$ channel uses, and $\rv{K}_{\rm a} \sim \mathrm{Poisson}(50)$.}
	\label{fig:Pe_vs_EbN0}
\end{figure}
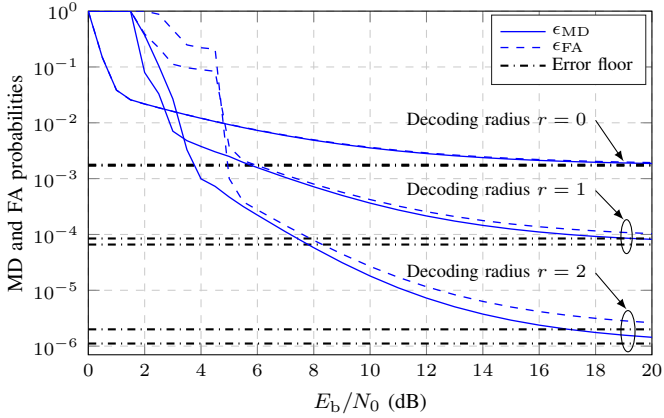

\section{Conclusions} \label{sec:conclusions}
We proposed a formulation for massive uncoordinated access where both the identity and the number of active users are unknown. We derived a random-coding bound for the Gaussian multiple access channel that reveals a trade-off between misdetection and false alarm. Our bound \revisee{provides an estimate of} the penalty in terms of energy efficiency due to the lack of knowledge of the number of active users, and serves as a benchmark to assess the performance of practical schemes. Possible future works include extending our bound to the \gls{MAC} with fading and multiple antennas.

\section*{Acknowledgement}

This work has been supported by the Wallenberg AI, Autonomous Systems and Software Program (WASP).



\appendices
\section{Proof of Theorem~\ref{thm:RCU_unknownKa}} \label{app:proof}
 The following well-known results will be used in our proof.
 
\begin{lemma}[{Change of measure~\cite[Lemma~4]{Ohnishi2020novel}}] \label{lem:change_measure}
	Let $p$ and $q$ be two probability measures. Consider a random variable $\rv{x}$ supported on $\Hc$ and a function $f \colon \Hc \to [0,1]$. It holds that 
	\begin{align}
		\E_p[f(\rv{x})] \le \E_q[f(\rv{x})] + d_{\rm TV}(p,q)
	\end{align}
	where $d_{\rm TV}(p,q)$ denotes the total variation distance between $p$ and $q$.
\end{lemma}

\begin{lemma}[{Chernoff bound~\cite[Th. 6.2.7]{DeGroot2012ProbStats}}] \label{lem:Chernoff}
	For a random variable $\rv{x}$ with moment-generating function $\E[e^{t \rv{x}}]$ defined for all $|t| \le b$, it holds for all $\lambda \in [0,b]$ that
	\begin{align}
		\P[\rv{x} \le x] \le e^{\lambda x} \E[e^{-\lambda \rv{x}}].
	\end{align}
\end{lemma}

\begin{lemma} [{Gallager's $\rho$-trick~\cite[p.~136]{Gallager1968information}}] \label{lem:Gallager}
	It holds that $\P[\cup_i A_i] \le (\sum_{i} \P[A_i])^\rho$ for every $\rho \in [0,1]$.
\end{lemma}
\begin{lemma} \label{lem:chi2}
	Let $\rvVec{x} \sim \Cc\Nc(\muv,\sigma^2\Id_n)$. For all $\gamma > -\frac{1}{\sigma^2}$, it holds that 
	\begin{align}
		\E[e^{-\gamma \|\rvVec{x}\|^2}] = (1+\gamma\sigma^2)^{-n} \exp\bigg(-\frac{\gamma\|\muv\|^2}{1+\gamma\sigma^2}\bigg). \label{eq:tmp363}
	\end{align} 
\end{lemma}
	
We present next an error analysis of the random-coding scheme introduced in Section~\ref{sec:RCU}. 
Denote by $\Wc_0$ the set of misdetected messages, i.e., $\Wc_0 \defeq \Wc \setminus \widehat{\Wc}$, and by $\Wc_0'$ the set of falsely alarmed messages, i.e., $\Wc_0' \defeq \widehat{\Wc} \setminus \Wc$. The \gls{MD} and \gls{FA} probabilities, defined respectively in~\eqref{eq:eps_MD} and \eqref{eq:eps_FA}, can be expressed as the average fraction of misdetected and falsely alarmed messages, respectively, i.e., 
\begin{align}
	P_{\rm MD} &= \E[\ind{|\Wc| \ne 0} \cdot \md], \label{eq:pMD}\\
	P_{\rm FA} &= \E[\ind{|\widehat{\Wc}| \ne 0} \cdot \fa]. \label{eq:pFA}
\end{align}

\subsection{A Change of Measure} \label{sec:change_measure}
Recall that $|\Wc|$ is the number of {\em distinct} transmitted messages. Since multiple transmitters may pick the same codeword to transmit, $|\Wc|$ can be smaller than $\rv{K}_{\rm a}$. Since both $\ind{|\Wc| \ne 0} \cdot \frac{|\Wc_0|}{|\Wc|}$ and $\ind{|\widehat{\Wc}| \ne 0} \cdot \frac{|\Wc_0'|}{|\widehat{\Wc}|}$ are nonnegative and upper-bounded by one, we can apply Lemma~\ref{lem:change_measure} to these random quantities. Specifically, we replace the measure over which the expectation is taken by the one under which: 
i) there are at least $K_{l}$ and at most $K_{u} \ge \overline{K_{\rm a}'}$ active users, i.e., $K_{l}\le\rv{K}_{\rm a} \le K_{u}$; ii) $\widetilde{\rv{w}}_1,\dots,\widetilde{\rv{w}}_{\rv{K}_{\rm a}}$ are sampled uniformly without replacement from $[M]$, i.e., $|\Wc| = \rv{K}_{\rm a}$;  iii) $\rvVec{x}_i = \cv_{\rv{w}_i}, \forall i$ (instead of $\rvVec{x}_i = \cv_{\rv{w}_i} \ind{\|\cv_{\rv{w}_i}\|^2 \le nP}$). 

It then follows from \cite[Eq. (41)]{Kowshik2020fundamental} that the total variation between the true measure and the new one is
upper-bounded by $\P[{\rv{K}_{\rm a} \notin [K_{l}:K_u]}] + \P[ |\Wc| < \rv{K}_{\rm a}] + \P[\overline{U}]$, where
$U \defeq \{\|\cv_{\rv{w}_i}\|^2 \le nP, \forall i \in [\rv{K}_{\rm a}] \}$ and $\overline{U}$ denotes the complement of~$U$.
We compute these probabilities 
as follows:
\begin{itemize}[leftmargin=*]
	\item To compute the first probability, we simply use that $\P[{\rv{K}_{\rm a} \notin [K_{l}:K_u]}] = 1 - \sum_{K_{\rm a} = K_{l}}^{K_{u}}P_{\rv{K}_{\rm a}}(K_{\rm a})$.
	
	\item Consider a given $\rv{K}_{\rm a} = K_{\rm a}$. Since $\widetilde{\rv{w}}_1,\dots,\widetilde{\rv{w}}_{K_{\rm a}}$ are drawn uniformly and independently from $[M]$, there are $M^{K_{\rm a}}$ possible $K_{\rm a}$-tuples. Among them, $\frac{M!}{(M-K_{\rm a})!}$ tuples have nonduplicate elements. Therefore, $\P[ |\Wc| = K_{\rm a} \cond \rv{K}_{\rm a} = K_{\rm a}] = \frac{M!}{(M-K_{\rm a})!} \frac{1}{M^{K_{\rm a}}}$. As a consequence, $\P[ |\Wc| < \rv{K}_{\rm a}] = 1 - \P[ |\Wc| = \rv{K}_{\rm a}] = 1- \E_{\rv{K}_{\rm a}}\Big[\frac{M!}{M^{\rv{K}_{\rm a}}(M-\rv{K}_{\rm a})!}\Big]$.\footnote{In~\cite{PolyanskiyISIT2017massive_random_access}, $\P[ |\Wc| \le K_{\rm a}]$ is upper-bounded by $\binom{K_{\rm a}}{2}/M$ using the union bound.} 
	
	\item The probability $\P[\overline{U}]$ can be finally evaluated as
	\begin{align}
		\P[\overline{U}] &= \E_{\rv{K}_{\rm a}}\Bigg[{\P[\bigcup_{i=1}^{\rv{K}_{\rm a}} \|\cv_{\rv{w}_i}\|^2 > nP]}\Bigg] \\
		&\le \E_{\rv{K}_{\rm a}} \Bigg[\sum_{i=1}^{\rv{K}_{\rm a}}{\P[\|\cv_{\rv{w}_i}\|^2 > nP]}\Bigg] \label{eq:tmp675}\\
		&= \E[\rv{K}_{\rm a}]  \frac{\Gamma(n,nP/P')}{\Gamma(n)}, \label{eq:tmp676}
	\end{align}
	where \eqref{eq:tmp675} follows from the union bound and \eqref{eq:tmp676} holds since $\|\cv_{\rv{w}_i}\|^2$ follows the Gamma distribution with shape $n$ and scale $P'$.
\end{itemize}
From the above calculations, we deduce that the total variation between the two measures is upper-bounded by $p_0$ defined in~\eqref{eq:p0}.
Applying Lemma~\ref{lem:change_measure} to the random quantities $\ind{|\widehat{\Wc}| \ne 0} \cdot \frac{|\Wc_0|}{|\Wc|}$ and $\ind{|\Wc| \ne 0} \cdot \frac{|\Wc_0'|}{|\widehat{\Wc}|}$, we consider implicitly the new measure from now on at a cost of adding $p_0$ to their original expectations. It remains to bound the \gls{MD} and \gls{FA} probabilities given in~\eqref{eq:pMD} and \eqref{eq:pFA}, respectively, under the new measure. For the sake of clarity, in Appendix~\ref{sec:special_case}, we shall prove a bound on $P_{\rm MD}$ and $P_{\rm FA}$ for a special case where i) $\rv{K}_{\rm a}$ and $\rv{K}'_{\rm a}$ are fixed and $r = 0$, i.e., there are always $K_{\rm a}$ users transmitting and the decoder always outputs a list of size $K'_{\rm a}$; ii) $K'_{\rm a} < \min\{K_{\rm a}, M-K_{\rm a}\}$. Then, in Appendix~\ref{sec:general_case}, we shall show how to extend the proof for the general case where $\rv{K}_{\rm a}$ and $\rv{K}'_{\rm a}$ are random and $r \ge 0$.

\subsection{A Special Case} \label{sec:special_case}
In the aforementioned special case, \eqref{eq:eps_MD} and \eqref{eq:eps_FA} become   
\begin{align}
	\epsilon_{\rm MD} &= \sum_{t = 0}^{K_{\rm a}'}\frac{t+K_{\rm a}-K_{\rm a}'}{K_{\rm a}} \min\{p_{t,t},q_{t,t}\} + p_0, \label{eq:eps_MD_simp}\\
	\epsilon_{\rm FA} &= \sum_{t = 0}^{K_{\rm a}'}  \frac{t}{K'_{\rm a}} \min\{p_{t,t}, q_{t,t}\} + p_0, \label{eq:eps_FA_simp}
\end{align}		
where $p_{t,t}$ and $q_{t,t}$ will be derived next. We next show that $\epsilon_{\rm MD}$ and $\epsilon_{\rm FA}$ are indeed upper-bounds of $P_{\rm MD}$ and $P_{\rm FA}$, respectively, in this special case.

Observe that since the decoded list size $K'_{\rm a}$ is smaller than the number of transmitted messages $K_{\rm a}$, at least $K_{\rm a} - {K_{\rm a}'}$ messages are misdetected by default, and there can be $t \in [0:K_{\rm a}']$ additional \glspl{MD} occurring during the decoding process. Due to symmetry, we can assume without loss of generality that $\Wc = [K_{\rm a}]$ and that the list of messages that are initially misdetected due to insufficient decoded list size is $\Wc_{01} = [{K}_{\rm a} - {K_{\rm a}'}]$.\footnote{Note that due to user's identity ambiguity, this does not imply that the messages from a set of specific users are always misdetected.} Furthermore, let $\Wc_{02} = \Wc_{0} \setminus \Wc_{01}$ denote the set of $t$ additional \glspl{MD}.  Note that $\Wc_{02}$ is a generic subset of $[K_{\rm a} - {K_{\rm a}'} + 1:K_{\rm a}]$. Note also that $t$ is the number of \glspl{FA}, i.e., $|\Wc_0'| = t$. 
The set of transmitted messages can thus be expressed as $\Wc = \Wc_{01} \cup \Wc_{02} \cup (\Wc \setminus \Wc_0)$, and the received signal is $\rvVec{y} = c(\Wc_{01}) + c(\Wc_{02}) + c(\Wc \setminus \Wc_0) + \rvVec{z}$. Since the messages in $\Wc_{01}$ are always misdetected, the best approximation of $\Wc$ that the decoder can produce is $\Wc_{02} \cup (\Wc \setminus \Wc_0)$. However, under the considered error event $\Wc \to \widehat{\Wc}$, the actual decoded list is $\Wc_{0}' \cup (\Wc \setminus \Wc_0)$. Therefore, $\Wc \to \widehat{\Wc}$ implies that $\|\rvVec{y} - c(\Wc_{0}') - c(\Wc \setminus \Wc_0)\|^2 < \|\rvVec{y} - c(\Wc_{02}) - c(\Wc \setminus \Wc_0))\|^2$, which is equivalent to
\begin{align}
	\|c(\Wc_{01}) + c(\Wc_{02})- c(\Wc_{0}') + \rvVec{z}\|^2 < \|c(\Wc_{01}) + \rvVec{z}\|^2. \label{eq:eventF_simp}
\end{align}
We denote by $F(\Wc_{01},\Wc_{02},\Wc_{0}')$ the set of $\{\Wc_{01},\Wc_{02},\Wc_{0}'\}$ that satisfy~\eqref{eq:eventF_simp}.

We now compute the expectations in~\eqref{eq:pMD} and \eqref{eq:pFA}. 
Recall that, under assumptions just stated, we have $|\Wc_0| =  t + K_{\rm a} - K'_{\rm a}$, $|\Wc_0'| = |\Wc_{02}| = t$, and $|\widehat{\Wc}| = K'_{\rm a}$. 
It follows from \eqref{eq:pMD} and \eqref{eq:pFA} that, after the change of measure in Appendix~\ref{sec:change_measure}, $P_{\rm MD}$ and $P_{\rm FA}$ can be bounded as
\begin{align}
	P_{\rm MD} &\le \sum_{t=0}^{K'_{\rm a}} \frac{t+K_{\rm a}-{K_{\rm a}'}}{K_{\rm a}} \P[|\Wc_{02}| = t] + p_0,  \label{eq:tmp850_simp}\\
	P_{\rm FA} &\le \sum_{t=0}^{K'_{\rm a}} \frac{t}{K_{\rm a}'} \P[|\Wc_{02}| = t] + p_0, \label{eq:tmp853_simp}
\end{align} 
Next, we proceed to bound $\P[|\Wc_{02}| = t]$. 
This is done following two approaches. The first approach is based on error exponent analyses, resulting in the term $p_{t,t}$ in~\eqref{eq:eps_MD_simp}. The second approach is a variation of the dependence testing (DT) bound \cite[Th.~17]{Polyanskiy2010}, resulting in $q_{t,t}$ in~\eqref{eq:eps_MD_simp}.

\subsubsection{The Error-Exponent-Based Approach}  \label{sec:1st_approach}
By writing the event $|\Wc_{02}| = t$ as the union of the pairwise error events $F(\Wc_{01},\Wc_{02},\Wc_{0})$, we have that 
\begin{align}
	&\P[{|\Wc_{02}| = t}] \notag \\
	&= \P[\bigcup_{\Wc_{02} \subset [K_{\rm a} - {K_{\rm a}'} + 1:K_{\rm a}] \atop |\Wc_{02}| = t} \bigcup_{\Wc_{0}' \subset [K_{\rm a}+1:M] \atop |\Wc_{0}'| = t} \!F(\Wc_{01},\Wc_{02},\Wc_{0}')]\!, \label{eq:tmp901_simp} 
\end{align}
Next, given $c(\Wc_{01})$, $c(\Wc_{02})$, and $\rvVec{z}$, it holds for every $\lambda > -\frac{1}{tP'}$ that
\begin{align}
	&\P[F(\Wc_{01},\Wc_{02},\Wc_{0}')] \notag \\
	&\le e^{\lambda \|c(\Wc_{01}) + \|\rvVec{z}\|^2} \notag \\
	&\quad \cdot \E_{c(\Wc_{0}')}\Big[e^{-\lambda \|c(\Wc_{01}) + c(\Wc_{02})- c(\Wc_{0}') + \rvVec{z}\|^2}\Big] \label{eq:tmp766_simp}\\
	&= e^{\lambda \|c(\Wc_{01}) + \rvVec{z}\|^2} (1+\lambda tP')^{-n} \notag \\
	&\quad \cdot \exp\bigg(-\frac{\lambda\|c(\Wc_{01}) + c(\Wc_{02}) + \rvVec{z}\|^2}{1+\lambda t P'}\bigg), \label{eq:tmp768_simp}
\end{align}
where \eqref{eq:tmp766_simp} follows from the Chernoff bound in Lemma~\ref{lem:Chernoff}, and \eqref{eq:tmp768_simp} follows by computing the expectation in~\eqref{eq:tmp766_simp} using Lemma~\ref{lem:chi2}.
Next, we apply Gallager's $\rho$-trick in Lemma~\ref{lem:Gallager} to get that, given $c(\Wc_{01})$, $c(\Wc_{02})$, and $\rvVec{z}$, it holds for every $\rho \in [0,1]$ that
\begin{align}
	&\P[\bigcup_{\Wc_{0}' \subset [K_{\rm a}+1:M] \atop |\Wc_{0}'| = t} F(\Wc_{01},\Wc_{02},\Wc_{0}')] \\
	&\le \binom{M-K_{\rm a}}{t}^\rho (1+\lambda tP')^{-n\rho} \notag \\
	&\quad \cdot \exp\Bigg(\lambda \rho \bigg(\|c(\Wc_{01}) + \rvVec{z}\|^2 -\frac{\|c(\Wc_{01}) \!+\! c(\Wc_{02}) \!+\! \rvVec{z}\|^2}{1+\lambda t P'}\bigg)\Bigg). \label{eq:tmp803_simp}
\end{align}
Taking the expectation over $c(\Wc_{02})$ using Lemma~\ref{lem:chi2}, we obtain for given $c(\Wc_{01})$ and $\rvVec{z}$ that
\begin{align}
	&\P[\bigcup_{\Wc_{0}' \subset [K_{\rm a}+1:M] \atop |\Wc_{0}'| = t} F(\Wc_{01},\Wc_{02},\Wc_{0}')] \notag \\
	&\le \binom{M-K_{\rm a}}{t}^\rho (1\!+\!\lambda tP')^{-n\rho} \Big(1+\frac{\lambda \rho t P'}{1\!+\!\lambda tP'}\Big)^{-n} \notag \\
	&\quad \cdot \exp\Bigg(\!\lambda\rho \bigg(1-\frac{1}{1+\lambda P't(1+\rho)}\bigg)\|c(\Wc_{01}) + \rvVec{z}\|^2\Bigg) \\
	&= \binom{M-K_{\rm a}}{t}^\rho \exp\left(b_0\|c(\Wc_{01})  + \rvVec{z}\|^2 - na_0\right), \label{eq:tmp811}
\end{align}
where $a_0$ and $b_0$ are given by taking $t' = t$ in~\eqref{eq:a} and \eqref{eq:b}, respectively. 
Now applying Gallager's $\rho$-trick again, we obtain that, for every $\rho_1 \in [0,1]$,
\begin{align}
	&\P[\bigcup_{\Wc_{02} \subset [K_{\rm a} - {K_{\rm a}'} + 1:K_{\rm a}] \atop |\Wc_{02}| = t} \bigcup_{\Wc_{0}' \subset [K_{\rm a}+1:M] \atop |\Wc_{0}'| = t} F(\Wc_{01},\Wc_{02},\Wc_{0}')] \notag \\
	&\le \binom{K_{\rm a}'}{t}^{\rho_1} \binom{M-K_{\rm a}}{t}^{\rho\rho_1} \label{eq:tmp797_simp}\notag \\
	&\quad \cdot \E[\exp\left(\rho_1 b_0\|c(\Wc_{01}) + \rvVec{z}\|^2 - n\rho_1 a_0\right)] \\
	&= \binom{K_{\rm a}'}{t}^{\rho_1} \binom{M-K_{\rm a}}{t}^{\rho\rho_1} e^{-n\rho_1 a_0} \big(1-\rho_1P_2b_0\big)^{-n}, \label{eq:tmp800_simp}
\end{align}
where the last equality follows by computing the expectation in~\eqref{eq:tmp797_simp} jointly over $c(\Wc_{01})$ and $\rvVec{z}$ with the help of Lemma~\ref{lem:chi2}, and $P_2= 1+(K_{\rm a} - K_{\rm a}')P'$. Finally, plugging the result into \eqref{eq:tmp901_simp}, we obtain 
\begin{align}
	&\P[|\Wc_{02}| = t] \notag \\
	&\le \binom{K_{\rm a}'}{t}^{\rho_1} \binom{M-K_{\rm a}}{t}^{\rho\rho_1} e^{-n\rho_1 a_0} \big(1-\rho_1P_2b_0\big)^{-n} \\
	&\defeq p_{t,t}. \label{eq:tmp1148_simp}
\end{align}

\subsubsection{The DT-Based Approach} \label{sec:2nd_approach}
Next, we present an alternative bound on $\P[{|\Wc_{02}| = t}]$. Consider the channel law $P_{\rvVec{y} \cond c(\Wc_{0}), c(\Wc \setminus \Wc_0)}$ with input $c(\Wc_{0})$ and output $\rvVec{y}$ where $|\Wc_{02}| = t$. The corresponding information density~\cite[Def. 17.1]{Polyanskiy2019lecture} 
is given by
\begin{align}
	&\imath_t(c(\Wc_{0});\rvVec{y} \cond c(\Wc \setminus \Wc_0)) \notag \\
	&= n \ln(1+(t+K_a-K_{\rm a}')P') + \frac{\|\rvVec{y} - c(\Wc \setminus \Wc_0)\|^2}{1+(t+K_a-K_{\rm a}')P'} \notag \\
	&\quad - \|\rvVec{y} - c(\Wc_0) - c(\Wc \setminus \Wc_0)\|^2. 
\end{align}
Notice that the event $F(\Wc_{01},\Wc_{02},\Wc_{0}')$ defined in~\eqref{eq:eventF_simp} is equivalent to $\{\imath_t(c(\Wc_{0}');\rvVec{y} \cond c(\Wc \setminus \Wc_0)) > \imath_t(c(\Wc_{02});\rvVec{y} \cond c(\Wc \setminus \Wc_0))\}.$ 
Let
\begin{align} \label{eq:def_It_1_simp}
	\rv{I}_t \defeq \min_{\Wc_{02} \subset [K_{\rm a} - {K_{\rm a}'} + 1:K_{\rm a}] \atop |\Wc_{02}| = t} \imath_t(c(\Wc_{02});\rvVec{y} \cond c(\Wc \setminus \Wc_0)).
\end{align}
For a fixed arbitrary $\gamma$, it follows that
\begin{align}
	&\P[{|\Wc_{02}| = t}] \notag \\
	&=\P[I_{t} \le \gamma]\P[{|\Wc_{02}| = t \;\big|\; I_{t} \le \gamma}] \notag \\
	&\quad + \P[I_{t} > \gamma]\P[{|\Wc_{02}| = t \;\big|\; I_{t} > \gamma}] \\
	&\le \P[I_{t} \le \gamma] + \P[{|\Wc_{02}| = t \;\big|\; I_{t} > \gamma}] \label{eq:tmp838_simp}\\
	&= \P[I_{t} \le \gamma] \notag \\
	&\quad+ \P\bigg[ \bigcup_{\Wc_{02} \subset [K_{\rm a} - {K_{\rm a}'} + 1:K_{\rm a}] \atop |\Wc_{02}| = t} \bigcup_{\Wc_{0}' \subset [K_{\rm a}+1:M] \atop |\Wc_{0}'| = t} \bigg. \notag \\ &\qquad \qquad \bigg. \big\{\imath_t(c(\Wc_{0}');\rvVec{y} \cond c(\Wc \setminus \Wc_0)) \bigg. \notag \\ &\qquad \qquad \bigg. > \imath_t(c(\Wc_{02});\rvVec{y} \cond c(\Wc \setminus \Wc_0))\big\} \;\big|\; I_{t} > \gamma\bigg] \label{eq:tmp814_simp} \\
	&\le  \P[I_{t} \le \gamma] \notag \\
	&\quad+ \P\bigg[\bigcup_{\Wc_{02} \subset [K_{\rm a} - {K_{\rm a}'} + 1:K_{\rm a}] \atop |\Wc_{02}| = t} \bigcup_{\Wc_{0}' \subset [K_{\rm a}+1:M] \atop |\Wc_{0}'| = t} \bigg. \notag \\ &\qquad \qquad \bigg. \big\{\imath_t(c(\Wc_{0}');\rvVec{y} \cond c(\Wc \setminus \Wc_0)) > \gamma\big\}\bigg]. \label{eq:tmp383_simp}
\end{align}
Here, \eqref{eq:tmp814_simp} follows by writing explicitly the event $\{|\Wc_{02}| = t\}$, and \eqref{eq:tmp383_simp} follows by relaxing the inequality inside the second probability.
Using that $\P[\imath(x;\rv{y}) > \gamma] \le e^{-\gamma}, \forall x$~\cite[Cor.~17.1]{Polyanskiy2019lecture}, we obtain
\begin{align}
	\P[\imath_t(c(\Wc_{0}');\rvVec{y} \cond c(\Wc \setminus \Wc_0)) > \gamma] \le e^{-\gamma}.
\end{align}
Then, by applying the union bound and taking the infimum over $\gamma$, we conclude that
\begin{align}
	&\P[{|\Wc_{02}| = t}] \notag \\
	&\le \inf_{\gamma} \Bigg( \P[\rv{I}_t \le \gamma] +  \binom{K_{\rm a}'}{t}  \binom{M-K_{\rm a}}{t}  e^{-\gamma} \Bigg) \\
	&\defeq q_{t,t}. \label{eq:tmp1201_simp}
\end{align}
This concludes the DT-based approach.

It follows from \eqref{eq:tmp1148_simp} and \eqref{eq:tmp1201_simp} that
$\P[|\Wc_{02}| = t] \le \min\left\{p_{t,t}, q_{t,t} \right\}$.
Introducing this bound into \eqref{eq:tmp850_simp} and \eqref{eq:tmp853_simp}, we obtain that the \gls{MD} and \gls{FA} probabilities, averaged over the Gaussian codebook ensemble, are upper-bounded by $\epsilon_{\rm MD}$ and $\epsilon_{\rm FA}$ given in~\eqref{eq:eps_MD_simp} and \eqref{eq:eps_FA_simp}, respectively. 


\subsection{The General Case}  \label{sec:general_case}
We now explain how the result in the special case considered in the previous subsection can be extended to the general case where $\rv{K}_{\rm a}$ and $\rv{K}'_{\rm a}$ are random and $r \ge 0$. 

For random $\rv{K}_{\rm a}$ and $\rv{K}'_{\rm a}$, one has to take into account all the possible combinations of the number of transmitted messages and decoded messages when computing the expectations in~\eqref{eq:pMD} and \eqref{eq:pFA}. Consider the event that $K_{\rm a}$ users are active and the estimation of $K_{\rm a}$ results in $K_{\rm a}'$, which we denote by $K_{\rm a} \to K_{\rm a}'$. As in the special case, we assume without loss of generality that $\Wc = [{K}_{\rm a}]$. Furthermore, due to symmetry, we let $\Wc_{01} = [({K}_{\rm a} - \overline{K_{\rm a}'})^+]$ denote the list of $(\rv{K}_{\rm a} - \overline{K_{\rm a}'})^+$ initial \glspl{MD} due to insufficient decoded list size, and $\Wc_{02} = \Wc_0 \setminus \Wc_{01}$ the $t$ additional \glspl{MD} occurring during the decoding process.  
Note also that,
if $\underline{K_{\rm a}'} > K_{\rm a}$, the decoder always outputs more than $K_{\rm a}$ messages. Hence, at least $\underline{K_{\rm a}'} - K_{\rm a}$ decoded messages are falsely alarmed. Due to symmetry, let $\Wc_{01}' = [\rv{K}_{\rm a} + 1: \underline{K_{\rm a}'}]$ denote the list of $(\underline{K_{\rm a}'}-\rv{K}_{\rm a})^+$ initial \glspl{FA} due to excessive decoded list size, and $\Wc_{02}' = \Wc_0' \setminus \Wc_{01}'$ the $t'$ additional \glspl{FA} occurring during the decoding process. See Fig.~\ref{fig:venn} for a diagram depicting the relation between these sets of messages. Under these assumptions, $\Wc_{02}$ and $\Wc_{02}'$ are generic subsets of $[({K}_{\rm a} - \overline{K_{\rm a}'})^+ + 1:{K}_{\rm a}]$  and $[\max\{{K}_{\rm a},\underline{K_{\rm a}'}\}+1 : M]$, respectively.

Note that in the special case considered in Appendix~\ref{sec:special_case}, $t$ can take value from $0$ to $K'_{\rm a}$, and $t' = t$. In the general case, instead:
\begin{itemize}
	\item The possible values of $t$ are given by $\Tc$ defined in~\eqref{eq:T}. This is because the number of \glspl{MD}, given by $t+{(K_{\rm a}-\overline{K_{\rm a}'})}^+$, is upper-bounded by the total number of transmitted messages $K_{\rm a}$, and by $M-\underline{K_{\rm a}'}$ (since at least $\underline{K_{\rm a}'}$ messages are returned).
	
	\item Given $t$, the integer $t'$ takes value in $\overline{\Tc}_t$ defined in~\eqref{eq:Tbart} because: i) the decoded list size, given by $K_{\rm a} - t - {(K_{\rm a} - \overline{K_{\rm a}'})}^+ + t' + {(\underline{K_{\rm a}'}-K_{\rm a})}^+$, must be in $[\underline{K_{\rm a}'} : \overline{K_{\rm a}'}]$; ii) the number of \glspl{FA}, given by $t'+ (\underline{K_{\rm a}'}-K_{\rm a})^+$, is upper-bounded by the number of messages that are not transmitted $M-K_{\rm a}$, and by the maximal number of decoded messages $\overline{K_{\rm a}'}$.
	
	\item If the decoded list size is further required to be strictly positive, 
	then $t'$ takes value in $\Tc_t$ defined in~\eqref{eq:Tt}.
\end{itemize}

Using the above definitions, the best approximation of $\Wc$ that the decoder can produce is $\Wc_{02} \cup (\Wc \setminus \Wc_0) \cup \Wc_{01}'$, while the actual decoded list, under $\Wc \to \widehat{\Wc}$, is $\Wc'_{02} \cup (\Wc \setminus \Wc_0) \cup \Wc_{01}'$. Therefore, $\Wc \to \widehat{\Wc}$ implies that $\|\rvVec{y} - c(\Wc'_{02}) - c(\Wc \setminus \Wc_0) - c(\Wc_{01}')\|^2 < \|\rvVec{y} - c(\Wc_{02}) - c(\Wc \setminus \Wc_0) - c(\Wc_{01}')\|^2$, which is equivalent to
\begin{multline}
	\|c(\Wc_{01}) + c(\Wc_{02})- c(\Wc_{01}') - c(\Wc_{02}') + \rvVec{z}\|^2 \\
	 < \|c(\Wc_{01}) - c(\Wc_{01}') + \rvVec{z}\|^2. \label{eq:eventF}
\end{multline}
We denote by $F(\Wc_{01},\Wc_{02},\Wc_{01}',\Wc_{02}')$ the set of $\{\Wc_{01},\Wc_{02},\Wc_{01}',\Wc_{02}'\}$ that satisfy~\eqref{eq:eventF}.

We now compute the expectations in $P_{\rm MD}$ and $P_{\rm FA}$. 
Given $|\Wc_{02}| = t$ and $|\Wc_{02}'| = t'$, we have that $|\Wc_0| = t+(\rv{K}_{\rm a} - \overline{K_{\rm a}'})^+$, $|\Wc_0'| = t + (\underline{K_{\rm a}'}-\rv{K}_{\rm a})^+$, and $|\widehat{\Wc}| = \rv{K}_{\rm a} - t - (\rv{K}_{\rm a} - \overline{K_{\rm a}'})^+ + t' + (\underline{K_{\rm a}'}-\rv{K}_{\rm a})^+$. 
It follows from \eqref{eq:pMD} and \eqref{eq:pFA} that, after the change of measure in Appendix~\ref{sec:change_measure}, $P_{\rm MD}$ and $P_{\rm FA}$ can be bounded as
\begin{align}
	P_{\rm MD} &\le \sum_{K_{\rm a} =\max\{K_{l},1\}}^{K_{u}} P_{\rv{K}_{\rm a}}(K_{\rm a}) \sum_{K_{\rm a}' = K_{l}}^{K_{u}} \sum_{t\in \Tc}\frac{t+(K_{\rm a}-\overline{K_{\rm a}'})^+}{K_{\rm a}} \notag \\
	&\qquad \quad \cdot \P[|\Wc_{02}| = t, K_{\rm a} \to K_{\rm a}'] \notag \\
	&\quad + p_0,  \label{eq:tmp850}\\
	P_{\rm FA} &\le \sum_{K_{\rm a} =K_{l}}^{K_{u}} P_{\rv{K}_{\rm a}}(K_{\rm a}) \sum_{K_{\rm a}' = K_{l}}^{K_{u}}  \sum_{t\in \Tc} \sum_{t' \in \Tc_t} \notag \\
	&\qquad \quad \frac{t+(\underline{K_{\rm a}'} - K_{\rm a})^+}{K_{\rm a} \!-\! t \!-\! {(K_{\rm a} \!-\! \overline{K_{\rm a}'})}^+ \!\!+\! t' \!+\! {(\underline{K_{\rm a}'}\!-\!K_{\rm a})}^+\!} \notag \\
	&\qquad \quad \cdot \P[|\Wc_{02}| = t, |\Wc_{02}'| = t', K_{\rm a} \to K_{\rm a}'] \notag \\
	&\quad + p_0. \label{eq:tmp853}
\end{align} 

Next, we proceed to bound the joint probabilities $\P[|\Wc_{02}| = t,K_{\rm a} \to K_{\rm a}']$ and $\P[|\Wc_{02}| = t,|\Wc_{02}'| = t',K_{\rm a} \to K_{\rm a}']$. Let
\begin{align} \label{eq:def_A}
	A(K_{\rm a},K_{\rm a}') \defeq 
	\{m(\rvVec{y},K_{\rm a}') < m(\rvVec{y},K), \forall K \ne K_{\rm a}'\}.
\end{align}
Since the event $K_{\rm a} \to K_{\rm a}'$ implies that $|\widehat{\Wc}| \in [\underline{K_{\rm a}'}:\overline{K_{\rm a}'}]$ and $A(K_{\rm a},K_{\rm a}')$,
we have 
\begin{align}
	&\P[|\Wc_{02}| = t, K_{\rm a} \to K_{\rm a}'] \notag \\
	&\le \P[{{|\Wc_{02}| = t, |\widehat{\Wc}|\in [\underline{K_{\rm a}'}:\overline{K_{\rm a}'}}],A(K_{\rm a},K_{\rm a}')}] \\
	&\le \min\left\{\P[{|\Wc_{02}| \!=\! t, |\widehat{\Wc}|\in [\underline{K_{\rm a}'}:\overline{K_{\rm a}'}]}],  \P[A(K_{\rm a},K_{\rm a}')] \right\}, \label{eq:tmp883}
\end{align}
where \eqref{eq:tmp883} follows from the fact that the joint probability is upper-bounded by each of the individual probabilities. 
Similarly, we can show that
\begin{align}
	&\P[|\Wc_{02}| = t,|\Wc_{02}'| = t', K_{\rm a} \to K_{\rm a}'] \notag \\
	&\le \min\Big\{\P[{|\Wc_{02}| = t, |\Wc_{02}'| = t', |\widehat{\Wc}|\in [\underline{K_{\rm a}'}:\overline{K_{\rm a}'}]}], \Big. \notag \\
	&\qquad\qquad\Big.   \P[A(K_{\rm a},K_{\rm a}')] \Big\}. \label{eq:tmp1054}
\end{align}
We next present the bounds on $\P[A(K_{\rm a},K_{\rm a}')]$, $\P[{|\Wc_{02}| \!=\! t,  |\widehat{\Wc}| \in [\underline{K_{\rm a}'}:\overline{K_{\rm a}'}]}]$, and $\P[{|\Wc_{02}| = t, |\Wc_{02}'| = t',  |\widehat{\Wc}| \in [\underline{K_{\rm a}'}:\overline{K_{\rm a}'}]}]$.

\subsubsection{Bound on $\P[A(K_{\rm a},K_{\rm a}')]$}
We have
\begin{align}
	\P[A(K_{\rm a},K_{\rm a}')] 
	&= \P[m(\rvVec{y},K_{\rm a}') < m(\rvVec{y},K), \forall K \ne K_{\rm a}'] \\ 
	&\le\min_{K\colon K \ne K_{\rm a}'}\P[m(\rvVec{y},K_{\rm a}') < m(\rvVec{y},K)] \\
	&= \xi(K_{\rm a},K_{\rm a}'), \label{eq:tmp1077}
\end{align}
where $\xi(K_{\rm a},K_{\rm a}')$ is given by~\eqref{eq:xi}, and \eqref{eq:tmp1077} holds since under the new measure, $\rvVec{y} \sim \Cc\Nc(\mathbf{0},(1+K_{\rm a} P')\Id_n)$ distribution. 

\subsubsection{Bounds of $\P[{|\Wc_{02}| = t,  |\widehat{\Wc}| \in [\underline{K_{\rm a}'}:\overline{K_{\rm a}'}]}]$} \label{sec:bound_tMDs}
As in Appendix~\ref{sec:special_case}, we follow two approaches to bound $\P[{|\Wc_{02}| = t,  |\widehat{\Wc}| \in [\underline{K_{\rm a}'}:\overline{K_{\rm a}'}]}]$. The first approach is based on error exponent analyses and the second approach is based on the DT bound. In the first approach, we write the event $\{|\Wc_{02}| = t, |\widehat{\Wc}| \in [\underline{K_{\rm a}'}:\overline{K_{\rm a}'}]\}$ as the union of the pairwise events and obtain
\begin{align}
	&\P[{|\Wc_{02}| = t, |\widehat{\Wc}| \in [\underline{K_{\rm a}'}:\overline{K_{\rm a}'}]}] \notag \\
	&= \P\Bigg(\bigcup_{t' \in \overline{\Tc}_t} \bigcup_{\Wc_{02} \subset [(K_{\rm a} - \overline{K_{\rm a}'})^+ + 1:K_{\rm a}] \atop |\Wc_{02}| = t}
	\bigcup_{\Wc_{02}' \subset [\max\{K_{\rm a},\underline{K_{\rm a}'}\}+1:M] \atop |\Wc_{02}'| = t'} \Bigg. \notag \\
	&\qquad \qquad \Bigg.F(\Wc_{01},\Wc_{02},\Wc_{01}',\Wc_{02}') \Bigg). \label{eq:tmp901} 
\end{align}
Then, by applying the Chernoff bound, Gallager's $\rho$-trick, and Lemma~\ref{lem:chi2}  following similar steps as in Appendix~\ref{sec:1st_approach}, we obtain
\begin{align}
	\P[{|\Wc_{02}| = t, |\widehat{\Wc}| \in [\underline{K_{\rm a}'}:\overline{K_{\rm a}'}]}] \le p_t \label{eq:tmp1148}
\end{align}
 with $p_t$ given by~\eqref{eq:pt}.  In the second approach, we consider the channel law $P_{\rvVec{y} \cond c(\Wc_{0}), c(\Wc \setminus \Wc_0)}$ with input $c(\Wc_{0})$ and output $\rvVec{y}$ where $|\Wc_{02}| = t$. The corresponding information density $\imath_t(c(\Wc_{0});\rvVec{y} \cond c(\Wc \setminus \Wc_0))$ is defined in~\eqref{eq:infor_den}. Notice that the event $F(\Wc_{01},\Wc_{02},\Wc_{01}',\Wc_{02}')$ defined in~\eqref{eq:eventF} is equivalent to $\{\imath_t(c(\Wc_{01}')+c(\Wc_{02}');\rvVec{y} \cond c(\Wc \setminus \Wc_0)) > \imath_t(c(\Wc_{01}') + c(\Wc_{02});\rvVec{y} \cond c(\Wc \setminus \Wc_0))\}.$ 
Then, by proceeding as in Appendix~\ref{sec:2nd_approach}, it follows that 
\begin{align}
	\P[{|\Wc_{02}| = t, |\widehat{\Wc}| \in [\underline{K_{\rm a}'}:\overline{K_{\rm a}'}]}] \le q_t \label{eq:tmp1201}
\end{align}
 with $q_t$ given by~\eqref{eq:qt}.

\subsubsection{Bounds of $\P[{|\Wc_{02}| = t, |\Wc_{02}'| = t', |\widehat{\Wc}| \in [\underline{K_{\rm a}'}:\overline{K_{\rm a}'}]}]$} 
First, we have that 
\begin{align}
	&\P[{|\Wc_{02}| = t, |\Wc_{02}'| = t', |\widehat{\Wc}| \in [\underline{K_{\rm a}'}:\overline{K_{\rm a}'}]}] \notag \\
	&= \P\Bigg[\bigcup_{\Wc_{02} \subset [(K_{\rm a} - \overline{K_{\rm a}'})^+ + 1:K_{\rm a}] \atop |\Wc_{02}| = t} 
	\bigcup_{\Wc_{02}' \subset [\max\{K_{\rm a},\underline{K_{\rm a}'}\}+1:M] \atop |\Wc_{02}| = t'} \bigg.\notag \\
	&\qquad \quad \bigg. F(\Wc_{01},\Wc_{02},\Wc_{01}',\Wc_{02}')\Bigg]. \label{eq:tmp365}
\end{align}
Notice that  $\P[{|\Wc_{02}| = t, |\Wc_{02}'| = t', |\widehat{\Wc}| \in [\underline{K_{\rm a}'}:\overline{K_{\rm a}'}]}]$ differs from  $\P[{|\Wc_{02}| = t, |\widehat{\Wc}| \in [\underline{K_{\rm a}'}:\overline{K_{\rm a}'}]}]$ in~\eqref{eq:tmp901} only in the absence of the union $\bigcup_{t'\in \overline{\Tc}_t}$. By applying the Chernoff bound, Gallager's $\rho$-trick, and Lemma~\ref{lem:chi2} following similar steps as in Appendix~\ref{sec:1st_approach}, we obtain that
\begin{align}
	\P[{|\Wc_{02}| = t, |\Wc_{02}'| = t', |\widehat{\Wc}| \in [\underline{K_{\rm a}'}:\overline{K_{\rm a}'}]}] \le  p_{t,t'} \label{eq:tmp1217}
\end{align}
with $p_{t,t'}$ given by~\eqref{eq:ptt}.
Alternatively, bounding $\P\Big[|\Wc_{02}| = t, |\Wc_{02}'| = t', |\widehat{\Wc}| \in [\underline{K_{\rm a}'}:\overline{K_{\rm a}'}]\Big]$ using the information density's property as in Appendix~\ref{sec:2nd_approach}, we obtain 
\begin{align}
	\P[{|\Wc_{02}| = t, |\Wc_{02}'| = t', |\widehat{\Wc}| \in [\underline{K_{\rm a}'}:\overline{K_{\rm a}'}]}] \le q_{t,t'} \label{eq:tmp1226}
\end{align}
 with $p_{t,t'}$ given by~\eqref{eq:qtt}.

\vspace{.3cm}
From \eqref{eq:tmp883}, \eqref{eq:tmp1077}, \eqref{eq:tmp1148}, and \eqref{eq:tmp1201}, we obtain that 
\begin{align}
	\P[|\Wc_{02}| = t, K_{\rm a} \to K_{\rm a}'] \le \min\left\{p_t, q_t, \xi(K_{\rm a},K_{\rm a}') \right\}.
\end{align}
From \eqref{eq:tmp1054}, \eqref{eq:tmp1077}, \eqref{eq:tmp1217}, and \eqref{eq:tmp1226}, we obtain that 
\begin{align}
	&\P[|\Wc_{02}| = t, |\Wc_{02}'| = t', K_{\rm a} \to K_{\rm a}'] \notag \\ &\le \min\left\{p_{t,t'}, q_{t,t'}, \xi(K_{\rm a},K_{\rm a}') \right\}.
\end{align}

Substituting these bounds on $\P[|\Wc_{02}| = t, K_{\rm a} \to K_{\rm a}']$ and $\P[|\Wc_{02}| = t, |\Wc_{02}'| = t', K_{\rm a} \to K_{\rm a}']$ into \eqref{eq:tmp850} and \eqref{eq:tmp853}, we deduce that the \gls{MD} and \gls{FA} probabilities, averaged over the Gaussian codebook ensemble, are upper-bounded by $\epsilon_{\rm MD}$ and $\epsilon_{\rm FA}$ given in~\eqref{eq:eps_MD} and \eqref{eq:eps_FA}, respectively. Finally, proceeding as in \cite[Th.~19]{Polyanskiy2011feedback}, one can show that there exists a randomized coding strategy that achieves \eqref{eq:eps_MD} and \eqref{eq:eps_FA} and involves time-sharing among three deterministic codes.

\section{Proof of Proposition~\ref{prop:xi}}
	\label{proof:xi}
	The \gls{PDF} of $\rvVec{y}_0$ is given by
	\begin{align}
		p_{\rvVec{y}_0}(\yv_0) = \frac{1}{\pi^n (1+K_{\rm a} P')^n} \exp\left(-\frac{\|\yv_0\|^2}{1+K_{\rm a} P'}\right).
	\end{align}
	Therefore, with the \gls{ML} estimation of $\rv{K}_{\rm a}$, we have that
	\begin{align}
		m(\rvVec{y}_0,K) &= -\ln p_{\rvVec{y}_0}(\yv_0) \\
		&= \frac{\|\yv_0\|^2}{1+K P'} + n\ln(1+K P') + n \ln \pi.
	\end{align}
	As a consequence, the event $m\left(\rvVec{y}_0,K_{\rm a}'\right) < m\left(\rvVec{y}_0,K\right)$ can be written as $\frac{\|\rvVec{y}_0\|^2}{1+K_{\rm a}' P'} + n\ln(1+K_{\rm a}' P') < \frac{\|\rvVec{y}_0\|^2}{1+K P'} + n\ln(1+K P')$, or equivalently, 
	\begin{equation}
		\|\rvVec{y}_0\|^2 \left(\frac{1}{1+K_{\rm a}'P'} - \frac{1}{1+KP'}\right) < n \ln\left(\frac{1+KP'}{1+K_{\rm a}'P'}\right). \label{eq:eventKa}
	\end{equation}
	Using the fact that $\|\rvVec{y}_0\|^2$ follows a Gamma distribution with shape $n$ and scale $1+K_{\rm a} P'$, we deduce that $\xi(K_{\rm a},K_{\rm a}')$ is given by~\eqref{eq:xi_ML} 	with $
	\zeta(K,K_{\rm a},K_{\rm a}')$ given by~\eqref{eq:zeta_ML}.
	
	For the energy-based estimation with $m(\yv,K) = |\|\yv\|^2 - n(1 + KP')|$, after some manipulations, the event  $m\left(\rvVec{y}_0,K_{\rm a}'\right) < m\left(\rvVec{y}_0,K\right)$ is equivalent to
	\begin{align}
		\begin{cases}
			\|\rvVec{y}_0\|^2 > n\left(1 + \frac{K_{\rm a} + K_{\rm a}'}{2}P'\right), &\text{if~} K_{\rm a}' < K_{\rm a}, \\
			\|\rvVec{y}_0\|^2 < n\left(1 + \frac{K_{\rm a} + K_{\rm a}'}{2}P'\right), &\text{if~} K_{\rm a}' > K_{\rm a}. 
		\end{cases}
	\end{align}
	Thus, from the Gamma distribution of $\|\rvVec{y}_0\|^2$, we deduce that $\xi(K_{\rm a},K_{\rm a}')$ is given by~\eqref{eq:xi_ML} with $\zeta(K,K_{\rm a},K_{\rm a}')$ given by~\eqref{eq:zeta_energy}.

\section{Error Floor Analysis} \label{app:error_floor}
For the decoder considered in Theorem~\ref{thm:RCU_unknownKa}, the initial \glspl{MD} and \glspl{FA} are unavoidable. On the other hand, the additional \glspl{MD} and \glspl{FA} can be reduced as the power $P$ increases. 
As $P\!\to\! \infty$, by assuming that no additional \gls{MD} or \gls{FA} occurs on top of these initial \glspl{MD} or \glspl{FA}, we obtain lower bounds on $\epsilon_{\rm MD}$ and $\epsilon_{\rm FA}$ as follows. 
\begin{proposition}[Asymptotic lower bounds on $\epsilon_{\rm MD}$ and $\epsilon_{\rm FA}$]
	With ML or energy-based estimation of $\rv{K}_{\rm a}$, 
	it holds that
	\begin{align}
		&\lim_{P\to\infty} \epsilon_{\rm MD} \notag \\ &\ge \bar{\epsilon}_{\rm MD} \notag \\
		&= \!\sum_{K_{\rm a} =\max\{K_\ell,1\}}^{K_{u}} \!\bigg(\!P_{\rv{K}_{\rm a}}(K_{\rm a}) \!\sum_{K_{\rm a}' = K_\ell}^{K_{u}} \! \frac{(K_{\rm a}\!-\!\overline{K_{\rm a}'})^+\!}{K_{\rm a}} 
		{\xi}(K_{\rm a},K_{\rm a}') \! \bigg) \!+\! \bar{p},\! \label{eq:eps_MD_floor}\\
		&\lim_{P\to\infty} \epsilon_{\rm FA} \notag \\  &\ge \bar{\epsilon}_{\rm FA} \notag \\
		&= \!\sum_{K_{\rm a} =K_\ell}^{K_{u}} \!\bigg(\!P_{\rv{K}_{\rm a}}(K_{\rm a}) \notag \\ 		& \qquad \cdot 
		\sum_{K_{\rm a}' = K_\ell}^{K_{u}} \! \frac{(\underline{K_{\rm a}'}-K_{\rm a})^+}{K_{\rm a} \!-\! {(K_{\rm a} \!-\! \overline{K_{\rm a}'})}^+ \!+\! {(\underline{K_{\rm a}'}\!-\!K_{\rm a})}^+} 
		{\xi}(K_{\rm a},K_{\rm a}') \! \bigg) \!+\! \bar{p}, \label{eq:eps_FA_floor}
	\end{align}	
	where $\bar{p} = 2 - \sum_{K_{\rm a} = K_\ell}^{K_{u}}P_{\rv{K}_{\rm a}}(K_{\rm a}) - \E_{\rv{K}_{\rm a}}\left[\frac{M!}{M^{\rv{K}_{\rm a}}(M-\rv{K}_{\rm a})!} \right]$, and $\xi(K_{\rm a},K_{\rm a}')$ is given by~\eqref{eq:xi_ML} with
	$\zeta(K,K_{\rm a},K_{\rm a}') = n \ln\big(\frac{K}{K_{\rm a}'}\big) K_{\rm a}^{-1}\big(\frac{1}{K_{\rm a}'} - \frac{1}{K}\big)^{-1}$ for ML estimation of $\rv{K}_{\rm a}$ and $\zeta(K,K_{\rm a},K_{\rm a}') = n\frac{K+K_{\rm a}'}{2 K_{\rm a}}$ for energy-based estimation of $\rv{K}_{\rm a}$. 
\end{proposition}
\begin{proof}
	First, the optimal value of $P'$ minimizing the bounds must grow with $P$ since otherwise $\tilde{p}$ will be large. Therefore, as $P\to \infty$, we can assume without loss of optimality that $P' \to \infty$. Next, when $t = t' = 0$, we can verify that $a = b = 0$, thus $E_0(\rho,\rho_1) = 0$ and $E(0,0) = 0$, achieved with $\rho = \rho_1 = 0$. Therefore, $p_0 = p_{0,0} = e^{-n\cdot 0} = 1$. We can also verify that $q_0$ and $q_{0,0}$ both converge to $1$ as $P' \to \infty$.
	When $P' \to \infty$, $\xi(K_{\rm a},K_{\rm a}')$ given in Proposition~\ref{prop:xi} converges to the right-hand side of \eqref{eq:xi_ML} with $\zeta(K,K_{\rm a},K_{\rm a}') = n \ln\big(\frac{K}{K_{\rm a}'}\big) K_{\rm a}^{-1}\big(\frac{1}{K_{\rm a}'} - \frac{1}{K}\big)^{-1}$ for ML estimation of $\rv{K}_{\rm a}$ and $\zeta(K,K_{\rm a},K_{\rm a}') = n\frac{K+K_{\rm a}'}{2 K_{\rm a}}$ for energy-based estimation of $\rv{K}_{\rm a}$. Furthermore, the last term in $\tilde{p}$ given by~\eqref{eq:p0} vanishes and thus $\tilde{p} \to \bar{p}$. Finally, the lower bounds $\bar{\epsilon}_{\rm MD}$ and $\bar{\epsilon}_{\rm FA}$ follows by substituting the asymptotic values of $p_0$, $q_0$, $p_{0,0}$, $q_{0,0}$, $\xi(K_{\rm a},K_{\rm a}')$, and $\tilde{p}$ computed above into $\epsilon_{\rm MD}$ and $\epsilon_{\rm FA}$, and by setting $\min\{p_t,q_t\}$ to zero for $t \ne 0$, and setting $\min\{p_{t,t'},q_{t,t'}\}$ to zero for $(t,t')\ne (0,0)$. 
%
%
%
%
	 
\end{proof}

We remark that the lower bounds in~\eqref{eq:eps_MD_floor} and~\eqref{eq:eps_FA_floor} are tight for typical IoT settings. Indeed, equalities in~\eqref{eq:eps_MD_floor} and~\eqref{eq:eps_FA_floor} hold if the probability of having additional \glspl{MD} and \glspl{FA} vanishes, i.e., $\min\{p_t,q_t\} \to 0$ for $t \ne 0$ and $\min\{p_{t,t'},p_{t,t'}\} \to 0$ for $(t,t') \ne (0,0)$ as $P\to\infty$. With $\rho = \rho_1 = 1$, 
the optimal $\lambda$ in~\eqref{eq:E0} is given by $\lambda \!=\! 1/(2P_2)$. Thus, by replacing the maximization over $\rho$ and $\rho_1$ in~\eqref{eq:Ett} with $\rho = \rho_1 = 1$, we obtain that $E(t,t') \ge -t' R_1 - R_2 + \ln\big(1+\frac{(t+t')P'}{4P_2}\big)$. It follows that 
\begin{align}
	p_{t,t'} &\le \binom{M\!-\!\max\{K_{\rm a}, \underline{K_{\rm a}'}\}}{t'} \binom{\min\{K_{\rm a}, \overline{K_{\rm a}'}\}}{t} \notag \\
	&\quad \cdot \left(1+\frac{(t+t')P'}{4P_2}\right)^{-n}.  \label{eq:bound_ptt} 
\end{align}
If $K_{\rm a} \in  [\underline{K_{\rm a}'}:\overline{K_{\rm a}'}]$, i.e., $P_2 = 1$, the right-hand side of~\eqref{eq:bound_ptt} vanishes as $P' \to\infty$. Otherwise, the right-hand side of~\eqref{eq:bound_ptt} converges to 
\begin{align}
\bar{p}_{t,t'} &= \binom{M-\max\{K_{\rm a}, \underline{K_{\rm a}'}\}}{t'} \binom{\min\{K_{\rm a}, \overline{K_{\rm a}'}\}}{t} \notag \\
&\quad \cdot \bigg(1+\frac{t+t'}{4((K_{\rm a} - \overline{K_{\rm a}'})^+ + (\underline{K_{\rm a}'} - K_{\rm a})^+)}\bigg)^{-n} \\
&\le M^{t'} K_{\rm a}^t \bigg(1+\frac{t+t'}{4((K_{\rm a} - \overline{K_{\rm a}'})^+ + (\underline{K_{\rm a}'} - K_{\rm a})^+)}\bigg)^{-n}.
\end{align}
Observe that $\bar{p}_{t,t'}$ is small if $n$ is relatively large compared to $\ln M$ and $\ln K_{\rm a}$, which is true for relevant values of $n,M$ and $K_{\rm a}$ in the IoT. Specifically, in typical IoT scenarios, $M$ and $K_{\rm a}$ are in the order of $10^2$, while $K_{\rm a}/n$ is from $10^{-4}$ to $10^{-3}$\textemdash see~\cite{PolyanskiyISIT2017massive_random_access} and \cite[Rem.~3]{Zadik2019}. For example, with $(M,n) = (2^{100}, 15000)$ and $K_{\rm a} \le 300$ as considered in~\cite{PolyanskiyISIT2017massive_random_access} and many follow-up works, assume that $(K_{\rm a} - \overline{K_{\rm a}'})^+ + (\underline{K_{\rm a}'} - K_{\rm a})^+ \le 20$, then $\bar{p}_{t,t'} < 10^{-128}$ for every $t\le 300$ and $t' \le 300$. 
As a consequence, $p_{t,t'}$ and $p_t$ are very small. We conclude that $\lim\limits_{P\to\infty} \epsilon_{\rm MD}$ and $\lim\limits_{P\to\infty} \epsilon_{\rm FA}$ approach closely $\bar{\epsilon}_{\rm MD}$ and $\bar{\epsilon}_{\rm FA}$, respectively. In other words, $\bar{\epsilon}_{\rm MD}$ and $\bar{\epsilon}_{\rm FA}$ essentially characterize the error floors of ${\epsilon}_{\rm MD}$ and ${\epsilon}_{\rm FA}$, respectively, as $P\to \infty$.

\bibliographystyle{IEEEtran}
\bibliography{IEEEabrv,./biblio}
\end{document}

%% file: preamble_global.tex
\usepackage{amsmath,amssymb}
\usepackage{xifthen}
\usepackage{mathtools}
\usepackage{enumerate}
\usepackage{microtype}
\usepackage{xspace}
\usepackage{bm}
\usepackage[T1]{fontenc}
\usepackage{fancyhdr}
\usepackage{lastpage}
\usepackage{bbm}

\newcommand{\mbs}[1]{\bm{#1}}
\newcommand{\vect}[1]{{\lowercase{\mbs{#1}}}}
\newcommand{\mat}[1]{{\uppercase{\mbs{#1}}}}

\newcommand{\T}{{\scriptscriptstyle\mathsf{T}}}

\newcommand{\cond}{\,\vert\,}
\renewcommand{\Re}[1][]{\ifthenelse{\isempty{#1}}{\operatorname{Re}}{\operatorname{Re}\left(#1\right)}}
\renewcommand{\Im}[1][]{\ifthenelse{\isempty{#1}}{\operatorname{Im}}{\operatorname{Im}\left(#1\right)}}

\newcommand{\cv}{\vect{c}}

\newcommand{\rv}{\vect{r}}

\newcommand{\yv}{\vect{y}}

\newcommand{\muv}{\vect{\mu}}

\newcommand{\Ac}{{\mathcal A}}

\newcommand{\Cc}{{\mathcal C}}

\newcommand{\Hc}{{\mathcal H}}

\newcommand{\Nc}{{\mathcal N}}

\newcommand{\Tc}{{\mathcal T}}
\newcommand{\Uc}{{\mathcal U}}
\newcommand{\Wc}{{\mathcal W}}

\newcommand{\Xc}{{\mathcal X}}
\newcommand{\Yc}{{\mathcal Y}}

\newcommand{\CC}{\mathbb{C}}

\newcommand{\Id}{\mat{\mathrm{I}}}

\newcommand{\CN}[1][]{\ifthenelse{\isempty{#1}}{\mathcal{N}_{\mathbb{C}}}{\mathcal{N}_{\mathbb{C}}\left(#1\right)}}

\renewcommand{\P}[1][]{\ifthenelse{\isempty{#1}}{\mathbb{P}}{\mathbb{P}\left[{#1}\right]}}
\newcommand{\E}[1][]{\ifthenelse{\isempty{#1}}{\mathbb{E}}{\mathbb{E}\left[#1\right]}}
\newcommand{\I}[1][]{\ifthenelse{\isempty{#1}}{\mathbb{I}}{\mathbb{I}\left\{#1\right\}}}
\renewcommand{\det}[1][]{\ifthenelse{\isempty{#1}}{\mathrm{det}}{\mathrm{det}\left(#1\right)}}
\newcommand{\trace}[1][]{\ifthenelse{\isempty{#1}}{{\rm tr}}{\mathrm{tr}\left(#1\right)}}
\newcommand{\rank}[1][]{\ifthenelse{\isempty{#1}}{\mathrm{rank}}{\mathrm{rank}\left(#1\right)}}
\newcommand{\diag}[1][]{\ifthenelse{\isempty{#1}}{\mathrm{diag}}{\mathrm{diag}\left(#1\right)}}
\newcommand{\Cov}[1][]{\ifthenelse{\isempty{#1}}{\mathsf{Cov}}{\mathsf{Cov}\left(#1\right)}}

\newcommand{\defeq}{\triangleq}

\newtheorem{proposition}{Proposition}
\newtheorem{definition}{Definition}
\newtheorem{theorem}{Theorem}
\newtheorem{corollary}{Corollary}
\newtheorem{lemma}{Lemma}

\newcounter{enumi_saved}
\usepackage{answers}
\Newassociation{solution}{Solution}{solutionfile}

\AtBeginDocument{\Opensolutionfile{solutionfile}[\jobname]}
\AtEndDocument{\Closesolutionfile{solutionfile}\clearpage
}

\IfFileExists{MinionPro.sty}{
}{
}
